\documentclass{article}

\author{Robert Ganian \and M. S. Ramanujan \and Stefan Szeider}

\date{\small Algorithms and Complexity Group, TU Wien, Vienna, Austria}

\usepackage{amsthm}

\newtheorem{theorem}{Theorem}
\newtheorem{corollary}{Corollary}
\newtheorem{lemma}{Lemma}
\newtheorem{definition}{Definition}
\newtheorem{remark}{Remark}

\newcommand{\lv}[1]{#1}
\newcommand{\sv}[1]{}

\sv{\title{Combining Treewidth and Backdoors for CSP}}
\lv{\title{Combining Treewidth and Backdoors for CSP\thanks{Research
      supported by the Austrian Science Funds (FWF), project P26696}}}

\usepackage[letterpaper,  margin=.7in]{geometry}
\usepackage{amsmath}
\usepackage{amsfonts}
\usepackage{graphicx,tikz}
\RequirePackage{fancyhdr}
\usepackage{xcolor}
\usepackage{boxedminipage}
\usepackage[inline]{enumitem}

\usepackage{todonotes}

\usepackage{amssymb,amsmath,amsthm}

\usepackage[linesnumbered, boxed, algosection]{algorithm2e}

\usepackage{amstext}
 
\theoremstyle{definition}
\newtheorem{proposition}{Proposition}

\usepackage{todonotes}
\usepackage{amssymb,amsmath,amsthm}
\usepackage{comment,tikz}
\usepackage{enumitem}
\usepackage{microtype}
\usetikzlibrary{shapes,arrows,fit,calc,positioning}
\usepackage{times}

\newcommand{\FPT}{\textsf{FPT}}
\newcommand{\sd}{\textsf{}}

\newtheorem{redrule}{Reduction Rule}	
\newtheorem{observation}{Observation}
\newtheorem{claim}{Claim}

\usepackage{boxedminipage}
\usepackage{comment}

\newcommand{\tw}{\textit{tw}}

\newcommand{\sbddetection}{{\sc Width Strong-$\CSP(\Gamma)$ Backdoor Detection}}

\newcommand{\torso}{{\bf Torso}}

\newcommand{\sbw}{{\bf btw}}

\usepackage{amsthm}

\newcommand{\defparproblem}[4]{
  \vspace{1mm}
\noindent\fbox{
  \begin{minipage}{0.96\textwidth}
  \begin{tabular*}{\textwidth}{@{\extracolsep{\fill}}lr} #1  & {\bf{Parameter:}} #3 \\ \end{tabular*}
  {\bf{Input:}} #2  \\
  {\bf{Objective:}} #4
  \end{minipage}
  }
  \vspace{1mm}
}

\newcommand{\naturals}{{\mathbb N}}

\newcommand{\bigoh}{\mathcal{O}}

\newcommand{\cV}{\mathcal{V}}

\newcommand{\cY}{\mathcal{Y}}

\newcommand{\cR}{\mathcal{R}}

\newcommand{\cP}{\mathcal{P}}

\newcommand{\cD}{\mathcal{D}}
\newcommand{\cF}{\mathcal{F}}
\newcommand{\cH}{\mathcal{H}}

\newcommand{\cT}{\mathcal{T}}

\newcommand{\cX}{\mathcal{X}}
\newcommand{\cM}{\mathcal{M}}

\newcommand{\sharpp}{$\text{\normalfont \#}$}

\newcommand{\sharpCSP}{\text{\normalfont \#CSP}}
\begin{document}
\maketitle
\thispagestyle{empty} 
%\enlargethispage*{10mm}
\begin{abstract}

We show that CSP is fixed-parameter tractable when parameterized by
the treewidth of a backdoor into any tractable CSP problem over a
finite constraint language. This result combines the two prominent
approaches for achieving tractability for CSP: (i)~by structural restrictions
on the interaction between the variables and the constraints and
(ii)~by language restrictions on the relations that can be used inside
the constraints. Apart from defining the notion of backdoor-treewidth and showing how backdoors of small treewidth can be used to efficiently solve CSP, our main technical contribution is a fixed-parameter algorithm that finds a backdoor of small treewidth.
%The biggest obstacle in this research direction is that it is not at all clear how to efficiently find such backdoors of small treewidth, and our main technical contribution lies
%Our main technical result is an algorithm that finds a backdoor of small treewidth.
\end{abstract}

%\pagebreak
\section{Introduction}
%!TEX root = split_main.tex

%\paragraph{Background.}

\newcommand{\SB}{\{\,}%
\newcommand{\SM}{\;{:}\;}%
\newcommand{\SE}{\,\}}%

\newcommand{\var}{\text{\normalfont var}}
\newcommand{\spbd}{\text{\normalfont split-bd-size}}

\newcommand{\rel}{\text{\normalfont rel}}
\newcommand{\fun}{\text{\normalfont fun}}

\newcommand{\relD}{\DDD^*}
\newcommand{\CCC}{\mathcal{C}}
\newcommand{\HHH}{\mathcal{H}}
\newcommand{\III}{\mathbf{I}}
\newcommand{\KKK}{\mathbf{K}}
\newcommand{\VVV}{\mathcal{V}}
\newcommand{\DDD}{\mathcal{D}}
\newcommand{\FFF}{\mathcal{F}}
\newcommand{\CSP}{\text{\normalfont CSP}}
\newcommand{\VCSP}{\text{\normalfont VCSP}}
\newcommand{\strongbds}{\textsc{SBD}($\CSP(\Gamma_1)\uplus\cdots\uplus\CSP(\Gamma_d)$)}

The Constraint Satisfaction Problem (CSP) is a central and generic
computational problem which provides a common framework for
many theoretical and practical applications \cite{HellNesetril08}. An
instance of CSP consists of a collection of variables that must be
assigned values subject to constraints, where each constraint is given
in terms of a relation whose tuples specify the allowed combinations
of values for specified variables. The problem was originally
formulated by Montanari~\cite{Montanari74}, and has been found
equivalent to the homomorphism problem for relational
structures \cite{FederVardi98} and the problem of evaluating
conjunctive queries on databases~\cite{Kolaitis03}. 
CSP is NP-complete in general, and identifying the classes of CSP instances which can be
solved efficiently has become a prominent research area in theoretical computer science~\cite{CarbonnelCooper16}. 

One of the most classical approaches in this area relies on exploiting
the \emph{structure} of how variables and constraints interact with each other,
most prominently in terms of the \emph{treewidth} of graph
representations of CSP instances. The first result in this line of
research dates back to 1985, when Freuder~\cite{Freuder85} observed
that CSP is polynomial-time tractable if the \emph{primal treewidth},
which is the treewidth of the \emph{primal graph} of the instance, is
bounded by a constant. A large number of related results 
%There is a rich literature of related results
on structural restrictions for CSP have been obtained to date (see, e.g.,
\cite{CohenJeavonsGyssens08,DalmauKolaitisVardi02,GottlobLeoneScarcello02b,Grohe07,Marx13,SamerSzeider10}).

The other leading approach used to show the tractability of constraint
satisfaction relies on \emph{constraint languages}. In this case, the
polynomially tractable classes are defined in terms of a tractable
\emph{constraint language} $\Gamma$, which is a set of relations that
can be used in the instance. A landmark result in this area is
Schaefer's celebrated Dichotomy Theorem for Boolean
CSP~\cite{Schaefer78} which says that for every constraint languge
$\Gamma$ over the Boolean domain, the corresponding CSP problem is
either NP-complete or solvable in polynomial time. Feder and
Vardi~\cite{FederVardi98} conjectured that such a dichotomy holds for
all finite 
constraint languages. Although the conjecture is still open it has
been proven true for many important special cases (see, e.g., 
\cite{Bulatov06,Bulatov11,BulatovKrokhinJeavons01,
CooperCohenJeavons94,Dalmau00,HellNesetril90}).

Tractability due to restrictions on the constraint language and
tractability due to restrictions in terms of the structure of the CSP instance are often
considered complementary: under structural restrictions the domain
language can be arbitrary, whereas under constraint language
restrictions the variables and constraints can interact arbitrarily.

One specific tool that is frequently used to build upon the constraint
language approach detailed above is the notion of \emph{backdoors}, which provides a means of relaxing celebrated results on
tractable constraint languages to instances which are `almost'
tractable.  In particular, this is done by measuring the size of a
\emph{strong backdoor}~\cite{WilliamsGomesSelman03} to a selected
tractable class,
%(usually defined using a constraint language), 
where a strong backdoor is a set of variables with the property that every
assignment of these variables results in a CSP instance in the
specified class. A natural way of defining such a class is to consider all CSP instances whose constraints use relations from a constraint language $\Gamma$, denoted by $\CSP(\Gamma)$.
The last couple of years have seen several new
results for CSP using this backdoor-based
approach~(see, e.g., 
\cite{CarbonnelCooperHebrard14,GanianRamanujanSzeider16,GaspersMisraOrdyniakSzeiderZivny14}.
In particular, the general aim of research in this direction is to
obtain so-called \emph{fixed-parameter} algorithms, i.e., algorithms
where the running time only has a polynomial dependence on the input
size and the exponential blow-up is restricted exclusively to the size
of the backdoor (the \emph{parameter}). Parameterized decision problems which
admit such an algorithm belong to the complexity class \FPT.

\begin{figure}[t]
\vspace{-0.5cm}
\begin{center}
  \includegraphics[trim={0 30pt 0 0},clip,height=120 pt, width=190 pt]{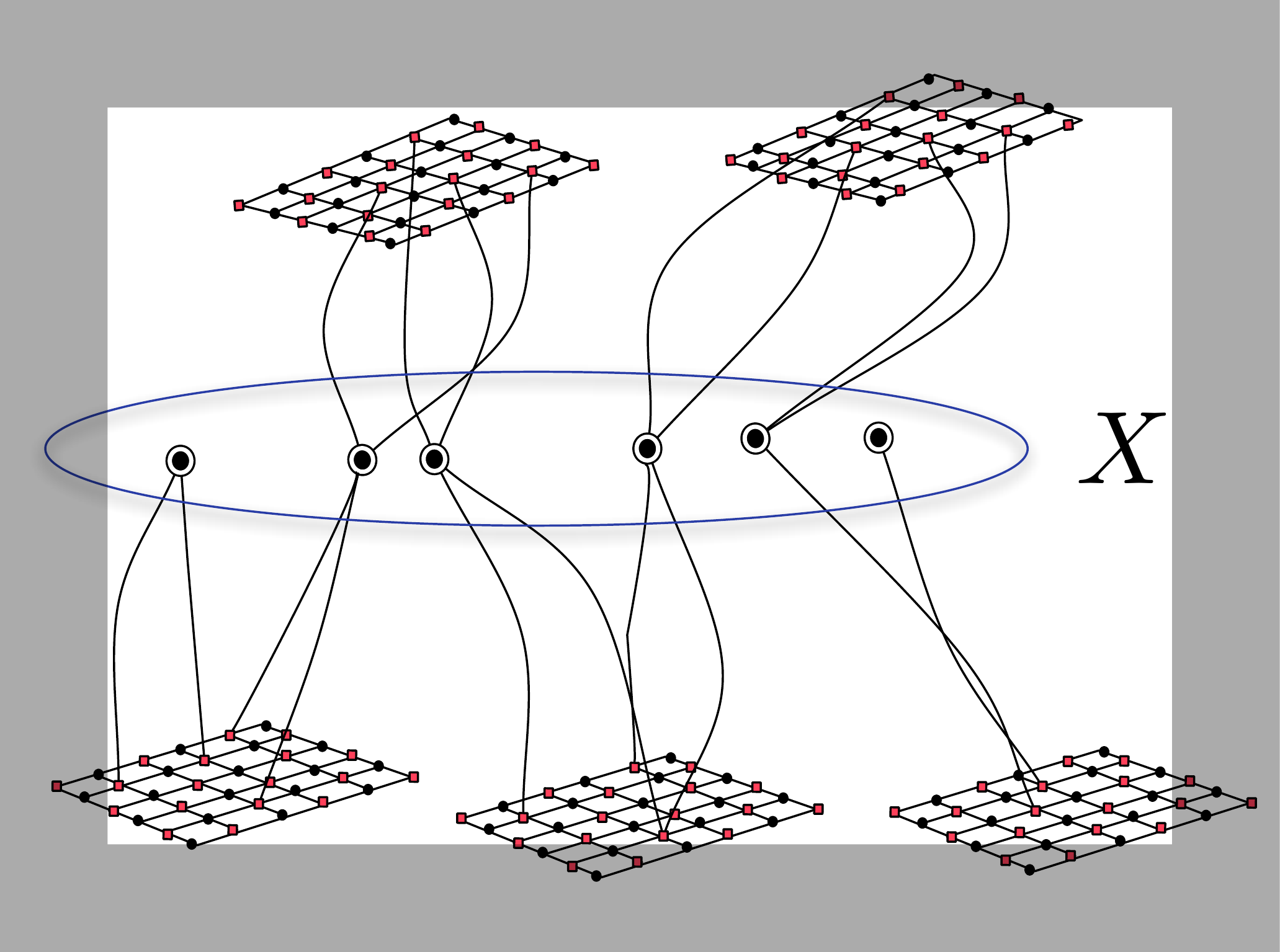}
\caption{An illustration of instances with neither a small backdoor into $\CSP(\Gamma)$ for any tractable constraint language $\Gamma$, nor bounded primal treewidth. Here, $X$ denotes a minimum strong backdoor of unbounded size into $\CSP(\Gamma)$ for some choice of $\Gamma$.} 
  %  \caption{An illustration of instances which have neither a small backdoor into any tractable constraint language, nor bounded primal treewidth. Here, $X$ denotes a minimum strong backdoor of unbounded size into a tractable constraint language.} 
  %  \caption{An illustration of instances which have neither a small backdoor nor bounded primal treewidth. Here, $X$ denotes a minimum strong backdoor set of unbounded size into the constraint language under consideration.}
\vspace{-0.4cm}
\label{fig:example} 
\end{center}
\vspace{-0.5cm}
\end{figure}

We note that treewidth-based and backdoor-based approaches outlined
above are orthogonal to each other. Consider, for example, on the one
hand a CSP instance which is tractable due to the used constraint
language but which has high treewidth, or on the other hand an
instance consisting of many disjoint copies of CSP instances of
constant primal treewidth with a constant-size strong backdoor into a tractable constraint language
(backdoor size multiplies whereas treewidth remains constant). Hence
applying either of these approaches (treewidth-based and
backdoor-based) alone will not yield satisfactory results for instances
that are not homogeneous with respect to either of these forms of restrictions. It is certainly natural to consider the algorithmic complexity of
instances which have small treewidth after certain simple `blocks'
characterized by a tractable constraint language are removed, or
instances with a large but `well-structured' backdoor to a
tractable class (see Fig.~\ref{fig:example}), but until now we lacked
the theoretical tools required to capture the complexity of
such instances.

\lv{
\subsection{Our Results}}
\sv{
\noindent{\bf Our Results.}}
We propose and develop a hybrid framework for constraint satisfaction which combines the advantages of both the width-based and backdoor-based approaches. In particular, we introduce the notion of \emph{backdoor-treewidth} with respect to a constraint language $\Gamma$; this is defined, roughly speaking, as the primal treewidth of the instance after contracting (possibly large) parts of the instance into single constraints, so that the remaining variables form a strong backdoor into $\CSP(\Gamma)$ in the original instance. We refer to Definition~\ref{def:width} for the formal definition of backdoor-treewidth. It is not difficult to see that backdoor-treewidth is at most the minimum of primal treewidth and the size of a backdoor into the specified class. However, backdoor-treewidth can be arbitrarily smaller than both the primal treewidth and the size of a backdoor, and hence promises to push the frontiers of tractability beyond the current state of the art.
%We propose and develop a hybrid framework for constraint satisfaction which combines the advantages of both the width-based and backdoor-based approaches. In particular, we introduce the notion of \emph{backdoor-treewidth} which is defined, roughly speaking, as the primal treewidth of the instance after contracting (possibly large) parts of the instance into single constraints, so that the remaining variables form a strong backdoor in the original instance. We refer to Definition~\ref{def:width} for the formal definition of backdoor-treewidth. It is not difficult to see that backdoor-treewidth is at most the minimum of primal treewidth and the size of a backdoor into the specified class. However, backdoor-treewidth can be arbitrarily smaller than both the primal treewidth and the size of a backdoor, and hence promises to push the frontiers of tractability beyond the current state of the art.

\begin{theorem}\label{thm:main main theorem}
  Let $\Gamma$ be a fixed tractable finite constraint language. Then, CSP  parameterized by the backdoor-treewidth with respect to $\Gamma$ is {\rm {\FPT}}.
\end{theorem}

We note that our result is in fact tight as far as the choice of the
language $\Gamma$ is concerned: $\Gamma$ must clearly be tractable,
and both the backdoor-based and width-based approaches are known to fail
for infinite languages under established complexity assumptions. To be
more specific, finding strong backdoors is not even \FPT\
parameterized by backdoor size if the arity of relations in the
language is unbounded~\cite{GaspersMisraOrdyniakSzeiderZivny14},
primal treewidth implicitly bounds the arity of relations, and both
approaches require bounded domain to
solve CSP in {\FPT} time~\cite{SamerSzeider10}.

%We need to deal with two separate problems in order to use backdoor-treewidth for solving constraint satisfaction:
Two separate problems need to be dealt with in order to use backdoor-treewidth for solving constraint satisfaction: 
finding a strong backdoor of small treewidth, and then using it to actually solve the CSP instance. The latter task can be solved efficiently by a dynamic programming procedure on a tree-decomposition. 
 However, finding strong backdoors of small treewidth is considerably more complicated and forms the main technical contribution of this article. We note in particular that algorithms for finding small backdoors to tractable classes cannot be used for this purpose, since the size of the backdoors we are interested in can be very large. In fact, it is even far from obvious that we can detect a backdoor of treewidth at most $k$ in polynomial time when $k$ is considered a constant (and the order of the polynomial may depend on $k$).
%for every fixed $k$. 
 
Our result on backdoor-treewidth also carries over to the counting variant of CSP ($\sharpCSP$).
  $\sharpCSP$ is a prominent \#P-complete
 extension of CSP problem which asks for the number of variable
 assignments that satisfy the given constraints. Structural
 restrictions as well as language restrictions have been studied for
 $\sharpCSP$. The dynamic programming algorithm for CSP for instances
 of bounded primal treewidth can be readily adapted to $\sharpCSP$
 (see, e.g., \cite{Farnqvist12}).  A constraint language $\Gamma$ is
 \emph{\#-tractable} if $\sharpCSP(\Gamma)$ ($\sharpCSP$ restricted to
 instances whose constraints use only relations from $\Gamma$) can be
 solved in polynomial time.  Bulatov~\cite{Bulatov13} characterized
 all finite \emph{\#-tractable} constraint languages.  Applying our results, we obtain the following corollary.

\begin{corollary}\label{cor:counting}
  Let $\Gamma$ be a fixed \#-tractable finite constraint language. Then, $\sharpCSP$ parameterized by the backdoor-treewidth with respect to $\Gamma$ is {\rm {\FPT}}.
\end{corollary}

\sv{

\noindent
{\bf (a)} In the first part of our algorithm to detect strong backdoors of small treewidth, we define a notion of \emph{boundaried} CSP instances in the spirit of boundaried graphs  and show that for any $t,k\in \mathbb N$, there is an equivalence relation $\sim_{t,k}$ on the set of all $t$-boundaried CSP instances such that (i) this relation has at most $f(k,t)$ equivalence classes for some function $f$ depending only on $k$ and the constraint language $\Gamma$,  and (ii) for any two $t$-boundaried CSP instances in the same equivalence class of $\sim_{t,k}$, they `interact in the same way' with every other $t$-boundaried CSP-instance.
\noindent
{\bf (b)} We then describe an algorithm that for any given $t,k\in \naturals$ runs in time $\bigoh(g(t,k))$ for some function $g$  and actually \emph{constructs} a set $\mathfrak{H}$ of $f(k,t)$ CSP instances, one from each equivalence class of the relation $\sim_{t,k}$. Additionally, we show that each instance in this set has size bounded by a function of $k$ and $t$.
\noindent
{\bf (c)} In this part, we show that for any given $t$-boundaried CSP instance $\III$ whose size exceeds a certain bound depending on $k$ and $t$ and whose incidence graph satisfies certain connectivity properties, we can in time $g(t,k)|\III|^{\bigoh(1)}$ correctly determine the equivalence class that this instance belongs to and compute a strictly smaller $t$-boundaried CSP instance $\III'$ which belongs to the same equivalence class of $\sim_{t,k}$ as $\III$. It follows that once $\III'$ is computed, we can `replace' $\III$ with the strictly smaller $\III'$, without altering the existence (or non-existence) of a strong backdoor of small treewidth. 
Our replacement framework is inspired by the graph replacement tools dating back to the results of Fellows and Langston~\cite{FellowsL89} and further developed by Arnborg, Bodlaender, and other authors~\cite{ArnborgCPS93,BodlaenderF96a,BodlaendervA01a,Fluiter97,BodlaenderH98}. 
%This leads to a reduction rule for reducing the size of the overall instance.
\noindent
{\bf (d)} In this part, we utilize the \emph{recursive-understanding} technique, introduced by Grohe et al. \cite{GroheKMW11} to solve the Topological Subgraph Containment problem and used with great success in the design of {\FPT} algorithms for several other fundamental graph problems (see \cite{KawarabayashiT11,ChitnisCHPP12}), to recursively compute a $t$-boundaried subinstance with the properties required to execute Part~\textbf{(c)}. Once this process terminates, we have an instance whose size is upper-bounded by a function of $k$ and $t$ which can be solved by brute force.

}

\lv{Our algorithm to detect strong backdoors of small treewidth has four parts. 

\begin{description}\item[(a)] In the first part, we define a notion of \emph{boundaried} CSP instances in the spirit of boundaried graphs  and show that for any $t,k\in \mathbb N$, there is an equivalence relation $\sim_{t,k}$ on the set of all $t$-boundaried CSP instances such that (i) this relation has at most $f(k,t)$ equivalence classes for some function $f$ depending only on $k$ and the constraint language $\Gamma$,  and (ii) for any two $t$-boundaried CSP instances in the same equivalence class of $\sim_{t,k}$, they `interact in the same way' with every other $t$-boundaried CSP-instance.

\item[(b)] We then describe an algorithm that for any given $t,k\in \naturals$ runs in time $\bigoh(g(t,k))$ for some function $g$  and actually \emph{constructs} a set $\mathfrak{H}$ of $f(k,t)$ CSP instances, one from each equivalence class of the relation $\sim_{t,k}$. Additionally, we show that each instance in this set has size bounded by a function of $k$ and $t$.

\item[(c)] In this part, we show that for any given $t$-boundaried CSP instance $\III$ whose size exceeds a certain bound depending on $k$ and $t$ and whose incidence graph satisfies certain connectivity properties, we can in time $g(t,k)|\III|^{\bigoh(1)}$ correctly determine the equivalence class that this instance belongs to and compute a strictly smaller $t$-boundaried CSP instance $\III'$ which belongs to the same equivalence class of $\sim_{t,k}$ as $\III$. It follows that once $\III'$ is computed, we can `replace' $\III$ with the strictly smaller $\III'$, without altering the existence (or non-existence) of a strong backdoor of small treewidth. 
Our replacement framework is inspired by the graph replacement tools dating back to the results of Fellows and Langston~\cite{FellowsL89} and further developed by Arnborg, Bodlaender, and other authors~\cite{ArnborgCPS93,BodlaenderF96a,BodlaendervA01a,Fluiter97,BodlaenderH98}.

%This leads to a reduction rule for reducing the size of the overall instance.

\item[(d)] In this part, we utilize the \emph{recursive-understanding} technique, introduced by Grohe et al. \cite{GroheKMW11} to solve the Topological Subgraph Containment problem and used with great success in the design of {\FPT} algorithms for several other fundamental graph problems (see \cite{KawarabayashiT11,ChitnisCHPP12}), to recursively compute a $t$-boundaried subinstance with the properties required to execute Part~\textbf{(c)}. Once this process terminates, we have an instance whose size is upper-bounded by a function of $k$ and $t$ which can be solved by brute force.

\end{description}
}

\lv{
\subsection{Related Work}}
\sv{
\noindent {\bf Related Work.}}
Williams et al.~\cite{WilliamsGomesSelman03,WilliamsGomesSelman03a}
introduced the notion of \emph{backdoors} for the runtime analysis of
algorithms for CSP and SAT, see also \cite{HemaspaandraWilliams12} for
a more recent discussion of backdoors for SAT.  A backdoor is a small
set of variables whose instantiation puts the instance into a fixed
tractable class (called the base class). One distinguishes between strong and weak backdoors,
where for the former all instantiations lead to an instance in the
base class, and for the latter at least one leads to a satisfiable
instance in the base class.  Backdoors have been studied under a
different name by Crama et al.~\cite{CramaEkinHammer97}. The study of
the parameterized complexity of finding small backdoors was initiated
by Nishimura et al.~\cite{NishimuraRagdeSzeider04-informal} for SAT,
who considered backdoors into the classes of Horn and Krom CNF
formulas. Further results cover the classes of renamable Horn
formulas~\cite{RazgonOSullivan09}, q-Horn
formulas~\cite{GaspersOrdyniakRamanujanSaurabhSzeider13} and classes
of formulas of bounded
treewidth~\cite{GaspersSzeider13,RamanujanLokshtanovFominSaurabhMisra15}. The
detection of backdoors for CSP has been studied in several works~\cite{BessiereCarbonnelHebrardKatsirelosWalsh13,CarbonnelCooperHebrard14}. Gaspers
et al.~\cite{GaspersMisraOrdyniakSzeiderZivny14} obtained
results on the detection of strong backdoors into \emph{heterogeneous}
base classes of the form $\CSP(\Gamma_1)\cup\dots \cup \CSP(\Gamma_d)$
where for each instantiation of the backdoor variables, the reduced
instance belongs entirely to some $\CSP(\Gamma_i)$ (possibly to
different $\CSP(\Gamma_i)$'s for different instantiations). This direction was recently further generalized by Ganian et al.~\cite{GanianRamanujanSzeider16} by developing a framework for detecting strong backdoors into so-called \emph{scattered} base classes with respect to $\Gamma_1\dots \Gamma_d$; there, each instantiation of the backdoor variables results in a reduced instance whose every connected component belongs entirely to some $\CSP(\Gamma_i)$ (possibly to different $\CSP(\Gamma_i)$'s for different components and different instantiations).

%The idea of
%graph replacement for algorithms dates back to Fellows and
%Langston~\cite{FellowsL89}. 
%Using this idea, Arnborg et
%al.~\cite{ArnborgCPS93} obtained a linear time algorithm for MSO expressible
%problems on graphs of bounded treewidth. Bodlaender and
%Fluiter~\cite{BodlaenderF96a,BodlaendervA01a,Fluiter97} generalized these
%ideas and applied it to optimization
%problems. Bodlaender and
%Hagerup~\cite{BodlaenderH98}, used the concept of graph reduction to
%obtain parallel algorithms for MSO expressible problems on bounded treewidth. 
\sv{\smallskip \noindent {\emph{Statements whose proofs are located in the appendix are marked with $\star$.}}}

\section{Preliminaries}\label{sect:prel}
We use standard graph terminology, see for instance the handbook by Diestel~\cite{Diestel12}. For $i\in \mathbb{N}$, we use $[i]$ to denote the set $\{1,\dots,i\}$.

\lv{
\subsection{Constraint Satisfaction}}
\sv{
\subparagraph*{Constraint Satisfaction.}}
\label{sub:CSP}
Let $\VVV$ be a set of variables and $\DDD$ a finite set of
values.  A \emph{constraint of arity $\rho$ over $\DDD$} is a pair $(S,R)$
where $S=(x_1,\dots,x_\rho)$ is a sequence of variables from $\VVV$ and
$R \subseteq \DDD^\rho$ is a $\rho$-ary relation. The set
$\var(C)=\{x_1,\dots,x_\rho\}$ is called the \emph{scope} of $C$.  An
\emph{assignment} $\alpha:X\rightarrow \DDD$ is a mapping of a set $X\subseteq \VVV$ of variables. An assignment
$\alpha:X\rightarrow \DDD$ \emph{satisfies} a constraint
$C=((x_1,\dots,x_\rho),R)$ if $\var(C)\subseteq X$ and
$(\alpha(x_1),\dots,\alpha(x_\rho)) \in R$.
For a set $\III$ of constraints we write $\var(\III)=\bigcup_{C\in \III}
\var(C)$ and $\rel(\III)=\SB R \SM (S,R) \in C, C\in \III \SE$.

A finite set $\III$ of constraints is \emph{satisfiable} if there
exists an assignment that simultaneously satisfies all the constraints
in $\III$.  The \emph{Constraint Satisfaction Problem} (CSP, for
short) asks, given a finite set $\III$ of constraints, whether $\III$
is satisfiable.  The \emph{Counting Constraint Satisfaction Problem}
(\sharpCSP, for short) asks, given a finite set $\III$ of constraints, to
determine the number of assignments to $\var(\III)$ that satisfy
$\III$. {\CSP} is NP-complete and {\sharpCSP} is {\sharpp}P-complete
(see, e.g., \cite{Bulatov13}).

Let $\alpha:X\rightarrow \DDD$ be an assignment. For a $\rho$-ary
constraint $C=(S,R)$ with $S=(x_1,\dots,x_\rho)$ and $R\in \DDD^\rho$, we denote by $C|_\alpha$
the constraint $(S',R')$ obtained from $C$ as follows. $R'$ is
obtained from $R$ by (i) deleting all tuples $(d_1,\dots,d_\rho)$ from
$R$ for which there is some $1\leq i \leq \rho$ such that $x_i\in X$ and $\alpha(x_i)\neq
d_i$, and (ii) removing from all remaining tuples all coordinates $d_i$
with $x_i\in X$.  $S'$ is obtained from $S$ by deleting all variables
$x_i$ with $x_i\in X$.  For a set $\III$ of constraints we define
$\III|_\alpha$ as $\SB C|_\alpha \SM C \in \III \SE$.  

A \emph{constraint language} (or \emph{language}, for short) $\Gamma$
over a domain $\DDD$ is a set of relations (of
possibly various arities) over~$\DDD$.  By $\CSP(\Gamma)$ we denote
CSP restricted to instances $\III$ with $\rel(\III)\subseteq
\Gamma$.  A constraint language is \emph{tractable} if for every
finite subset $\Gamma'\subseteq \Gamma$, the problem $\CSP(\Gamma')$
can be solved in polynomial time.  A constraint language is
\emph{\sharpp-tractable} if for every finite subset $\Gamma'\subseteq
\Gamma$, the problem $\sharpCSP(\Gamma')$ can be solved in polynomial
time. Throughout this paper, we make the technical assumption that every considered tractable
or \sharpp-tractable constraint language $\Gamma$ contains the redundant tautological relation of arity $2$; note that if this is not the case, then this relation can always be added into $\Gamma$ and the resulting language will still be tractable or \sharpp-tractable, respectively.
Let $\Gamma$ be a constraint language and $\III$ be an instance of CSP. A variable set $X$ is a \emph{strong backdoor} to $\CSP(\Gamma)$ if for each assignment $\alpha:X\rightarrow \DDD$ it holds that $\III|_\alpha\in \CSP(\Gamma)$. 

The \emph{primal graph} of a CSP instance $\III$ is the graph whose vertices correspond to the variables of $\III$ and where two variables $a,b$ are adjacent iff there exists a constraint in $\III$ whose scope contains both $a$ and $b$.
The \emph{incidence graph} of a CSP instance
$\III$  is the bipartite graph whose vertices correspond to the variables and constraints of $\III$, and where vertices corresponding to a variable $x$ and a constraint $C$ are adjacent if and only if $x\in \var(C)$. %We remark that we will use incidence graphs instead of primal graphs in our exposition for presentation reasons. 
Observe that an incidence graph does not uniquely define a CSP instance; however, in this paper the CSP instance from which a graph is obtained will always be clear from the context. Hence for an incidence or primal graph $G$% obtained from a CSP instance $\III$, 
we will denote the corresponding CSP instance by $\psi(G)$. Furthermore, we slightly abuse the notation and use $\cV(G)$ to refer to the vertices of $G$ that correspond to variables in $\psi(G)$, and $\CCC(G)$ to refer to the vertices of $G$ that correspond to constraints in $\psi(G)$. Also, for a vertex subset $X \subseteq V(G)$, we continue to use the notations $\cV(X)$ and $\CCC(X)$ to refer to the sets $\cV(G) \cap X$ and $\CCC(G) \cap X$, respectively.

%\begin{remark}
%Note that under the stated definition of incidence graphs, it is not true that any incidence graph uniquely defines a CSP instance. However, in this paper, the CSP instance from which the incidence graph $G$ is obtained will always be clear from the context and hence we may safely use $\psi(G)$ to refer to the CSP instance.
%\end{remark}

%
%\subsection{Parameterized Complexity}
%\label{sub:parcomp}
%A parameterized problem $\cP$ is a problem whose instances are tuples
%$(I,k)$, where $k\in \NN$ is called the parameter. We say that a
%parameterized problem is \emph{fixed parameter tractable} ({\FPT} in
%short) if it can be solved by an algorithm which runs in time
%$f(k)\cdot |I|^{\bigoh(1)}$ for some computable function~$f$;
%algorithms with running time of this form are called FPT algorithms.
%The notions of \emph{$\W[i]$-hardness} (for $i\in \NN$) are frequently
%used to show that a parameterized problem is not likely to be {\FPT}; an
%{\FPT} algorithm for a $\W[i]$-hard problem would imply that the
%Exponential Time Hypothesis
%fails~\cite{ChenHuangKanjXia06}. %Furthermore, there exists an even stronger notion of hardness for parameterized problems called \emph{para-\textup{NP}-hardness}, which also excludes algorithms running in time $|I|^{f(k)}$ unless P$=$NP.
%%We use the parameterized complexity paradigm of Downey and Fellows to
%%provide a more detailed complexity analysis of our results. 
%We refer the reader to other sources
%\cite{DowneyFellows99,DowneyFellows13,FlumGrohe06} for an in-depth
%introduction into parameterized complexity.
\lv{
\subsection{Treewidth}}
\sv{
\noindent{\bf Treewidth.}}
\label{trewap}
Let $G$ be a graph. A {\em tree decomposition} of $G$ is a pair $(T,\mathcal{X}=\{X_{t}\}_{t\in V(T)})$ where $T$ is a tree and ${\cal X}$ is a collection of subsets of $V(G)$ such that: 
\lv{
\begin{itemize}
\item $\forall {e=uv \in E(G)}, \exists {t\in V(T)} : \{u,v\} \subseteq X_{t}$, and 
\item $\forall {v\in V(G)}, \ T[\{t\mid v\in X_{t}\}]$ is a non-empty connected subtree of $T$. 
\end{itemize}
}
\sv{
\begin{enumerate*}[label={\textbf{(\arabic*)}}]
\item $\forall {e=uv \in E(G)}$, $\exists {t\in V(T)} : \{u,v\}$ $\subseteq X_{t}$, and 
\item $\forall {v\in V(G)}, \ T[\{t\mid v\in X_{t}\}]$ is a non-empty connected subtree of $T$. 
\end{enumerate*}
}
We call the vertices of $T$ {\em nodes} and the sets in $\mathcal{ X}$ {\em bags} of the tree decomposition $(T,{\cal X})$.  The {\em width} of $(T,{\cal X})$ is equal to $\max_{}\{|X_t|-1\mid {t\in V(T)}\}$ and the {\em treewidth} of $G$ (denoted $\tw(G)$) is the minimum width over all tree decompositions of $G$. 
\lv{

A {\em nice tree decomposition} is a pair $(T,{\cal X})$ where $(T,{\cal X})$ is a tree decomposition such that $T$ is a rooted tree and the following conditions are satisfied: 
\begin{itemize}
\item Every node of the tree $T$ has at most two children; 
\item if a node $t$ has two children $t_1$ and $t_2$, then $X_t = X_{t_1} = X_{t_2}$; and
\item  if a node $t$ has one child $t_1$, then either $|X_t| = |X_{t_1}| + 1$  and $X_{t_1} \subset X_{t}$ (in this case we call $t$ an {\em insert node}) or $|X_t| = |X_{t_1}| -1$ and $X_t \subset X_{t_1}$ (in this case we call $t$ a {\em forget node}). 
\end{itemize}
%\sv{
%\begin{enumerate*}[label={\textbf{(\arabic*)}}]
%\item Every node of the tree $T$ has at most two children,
%\item if a node $t$ has two children $t_1$ and $t_2$, then $X_t = X_{t_1} = X_{t_2}$, and
%\item  if a node $t$ has one child $t_1$, then either $|X_t| = |X_{t_1}| + 1$  and $X_{t_1} \subset X_{t}$ (and we call $t$ an {\em insert node}) or $|X_t| = |X_{t_1}| -1$ and $X_t \subset X_{t_1}$ (and we call $t$ a {\em forget node}). 
%\end{enumerate*}
%}

It is possible to transform a tree decomposition $(T,{\cal X})$ into a nice tree decomposition $(T',{\cal X}')$ in time $O(|V|+|E|)$~\cite{Bodlaender96}.
}
The primal treewidth of a $\CSP$ instance $\III$ is the treewidth of its primal graph, and similarly the incidence treewidth of $\III$ is the treewidth of its incidence graph. We note that if the constraints have bounded arity, then any class of $\CSP$ instances has bounded primal treewidth if and only if it has bounded incidence treewidth~\cite{SamerSzeider10a}.

\begin{proposition}[\cite{KolaitisVardi00}] \label{prop:treewidth_torso_bound}
Let $\III$ be a $\CSP$ instance where the constraints have arity bounded by $\rho\in \naturals$. Then, the primal treewidth of the instance is at most $\rho(t+1)-1$ where $t$ is the incidence treewidth of the instance.\end{proposition}

\lv{
\subsection{$t$-boundaried CSP Instances}}
\sv{
\noindent {\bf $t$-boundaried CSP Instances.}}
A \emph{$t$-boundaried graph} is a graph $G$  with a set $Z \subset V(G)$  of size at most $t$ with each vertex $v \in Z$  having a unique label $\ell(v) \in \{1, \ldots, t\}$. We refer to $Z$ as the {\em boundary} of $G$. For a $t$-boundaried graph $G$, $\delta(G)$ denotes the boundary of $G$. When it is clear from the context, we will often use the notation $(G,Z)$ to refer to a $t$-boundaried graph $G$ with boundary $Z$. For $P\subseteq [t]$, we use $P(G,Z)$ to denote the subset of $Z$ with labels in $P$; for $i\in [t]$ we use $i(G,Z)$ instead of $\{i\}(G,Z)$ for brevity. %We also interchangeably use the terms CSP instance and incidence graph.
Two $t$-boundaried graphs $G_1$ and $G_2$ can be `glued' together to obtain a new incidence graph, which we denote by $G_1 \oplus G_2$. The gluing operation takes the disjoint union of $G_1$ and $G_2$ and identifies the vertices of $\delta(G_1)$ and $\delta(G_2)$ with the same label. 

In some cases, we will also use a natural notion of replacement of boundaried graphs.
Let $(G_1,Z_1)$ be a $t$-boundaried graph which is an induced subgraph of a graph $G$ such that $Z_1$ is a separator between $V(G_1)\setminus Z_1$ and $V(G)\setminus V(G_1)$. Let $(G_2,Z_2)$ be a $t$-boundaried graph. Then the operation of \emph{replacement} of $(G_1,Z_1)$ by $(G_2,Z_2)$ results in the graph $G'=(G[V(G)\setminus (V(G_1)\setminus Z_1)],Z_1) \oplus (G_2,Z_2)$. Furthermore, if $G$ was a $j$-boundaried graph with boundary $Z$ and $Z\cap V(G_1)\subseteq Z_1$, then the resulting graph $G'$ is also a $j$-boundaried graph with the same boundary.

In this paper, it will sometimes be useful to lift the notions of boundaries and gluing from graphs to CSP instances. A \emph{$t$-boundaried incidence graph} of a CSP instance $\III$ is a $t$-boundaried graph $G$ with boundary $Z$ such that $G$ is the incidence graph of $\III$ and $Z\subseteq \cV$. Similarly, we call a CSP instance $\III$ with $t$ uniquely labeled variables a \emph{$t$-boundaried CSP instance}. Note that boundaried incidence graphs and boundaried CSP instances are de-facto interchangeable, but in some cases it is easier to use one rather than the other due to technical reasons. 

The gluing operations of boundaried incidence graphs and boundaried CSP instances are defined analogously as for standard boundaried graphs. Observe that if $G_1$ and $G_2$ are $t$-boundaried incidence graphs of $\III_1$ and $\III_2$, respectively, then $G_1\oplus G_2$ is also an incidence graph; furthermore, $\psi(G_1\oplus G_2)$ is well-defined and can be reconstructed from $\III_1$ and $\III_2$.

\lv{
\subsection{Minors}
\label{sub:minors}
Given an edge $e=xy$ of a graph $G$, the graph  $G/e$ is obtained from  $G$ by contracting the edge $e$, that is, the endpoints $x$  and $y$ are replaced by a new vertex $v_{xy}$ which  is  adjacent to the old neighbors of $x$ and $y$ (except from $x$ and $y$).  A graph $H$ obtained by a sequence of edge-contractions is said to be a \emph{contraction} of $G$.  We denote it by $H\leq_{c} G$. A graph $H$ is a {\em minor} of a graph $G$ if $H$ is the contraction of some subgraph of $G$ and we denote it by $H\leq_{m} G$. We say that a graph $G$ is {\em $H$-minor-free} when it does not contain $H$ as a minor. We also say that a graph class ${\cal G}$ is {\em $H$-minor-free} (or, excludes $H$ as a minor) when all its members are $H$-minor-free.

\begin{definition}\label{def:lgi}
Let $G_1$ and $G_2$ be two (not necessarily boundaried) graphs, and let $\Sigma$ be a set of symbols. For $i \in \{1,2\}$, let $f_{G_i}$ be a function that associates with every vertex of $V(G_i)$ some subset of $\Sigma$. The image of a vertex $v \in G_i$ under $f_{G_i}$ is called the label of that vertex. We say that that $G_1$ is {\em label-wise isomorphic} to $G_2$, and denote it by $G_1 \cong_{t} G_2$, if there is a map $h:V(G_1) \rightarrow V(G_2)$ such that (a) $h$ is a bijection; (b) $(u,v)\in E(G_1)$ if and only if $(h(u),h(v))\in E(G_2)$ and (c) $f_{G_1}(v)=f_{G_2}(h(v))$. We call $h$ a {\em label-preserving isomorphism}.
\end{definition}

Notice that the first two conditions of Definition~\ref{def:lgi} simply indicate that  $G_1$ and $G_2$ are isomorphic. Now, let $G$ be a $t$-boundaried graph, that is, $G$ has $t$ distinguished vertices, uniquely labeled from $1$ to $t$. Given  a $t$-boundaried graph $G$, we define a canonical labeling function $\mu_G: V(G)\rightarrow 2^{[t]}$. The function  $\mu_G$  maps every distinguished vertex $v$ with label $\ell \in [t]$ to the set $\{\ell\}$, that is $\mu_G(v)=\{\ell\}$, and for all vertices $v\in (V(G)\setminus \partial(G))$ we have that $\mu_G(v)=\emptyset$. 

Next we define a notion of labeled edge contraction. 
Let $H$ be a (not necessarily boundaried) graph together with a function $f_H:V(H)\rightarrow 2^{[t]}$ for some $t\in \naturals$ and $(u,v)\in E(H)$. Furthermore, let $H'$ be the graph obtained 
from $H$ by identifying the vertices $u$ and $v$ into $w_{uv}$, removing all the parallel edges and removing all the loops.  Then by 
{\em labeled edge contraction} of an edge $(u,v)$ of a graph $H$, we mean obtaining a graph $H'$ with the label function 
$f_{H'}:V(H')\rightarrow 2^{[t]}$. For  $x\in V(H')\cap V(H)$ we have that $f_{H'}(x)=f_H(x)$ and for $w_{uv}$ we define 
$f_{H'}(w_{uv})=f_H(u)\cup f_H(v)$.  Now we recall the notion of labeled minors of a 
$t$-boundaried graph. 

\begin{definition}
Let $H$ be a graph together with a function $f:V(H)\rightarrow 2^{[t]}$ and $G$ be a 
 $t$-boundaried graph with canonical labeling function $\mu_G$. 
A graph $H$ is called a labeled minor of $G$, if we can obtain a labeled isomorphic copy of $H$ from $G$ 
by performing edge deletion, vertex deletion and labeled edge contraction.
\end{definition}

\begin{remark}We note that the notion of a label-preserving isomorphism for graphs depends only on the labeling function, and is oblivious to the boundary. In particular, if $G$ and $H$ are two labeled $t$-boundaried graphs that are label-wise isomorphic, a label preserving isomorphism is not required to necessarily map the boundary vertices of $G$ to boundary vertices of $H$.  
\end{remark}

Finally, we define the notion of {\em $h$-folios} for boundaried graphs.

\begin{definition}
For $h\in \naturals$, the $h$-folio of a labeled graph $G$ with labeling $\Lambda$, is the set ${\cM}_h(G,\Lambda)$ of all labeled minors of $G$ on at most $h$ vertices. 
\end{definition}
}

\section{Backdoor-Treewidth}
\label{sec:backdoorwidth}
In this section we give a formal definition of the notion of backdoor-treewidth.

\begin{definition}
Let $G$ be a graph and $X\subseteq V(G)$.	We denote by $\torso_G(X)$ the following graph defined over the vertex set $X$. For every pair of vertices $x_1,x_2 \in X$, we add the edge $(x_1,x_2)$ if (a) $(x_1,x_2)\in E(G)$ or (b) $x_1$ and $x_2$ both have a neighbor in the same  connected component of $G-X$. That is, we begin with $G[X]$ and  make the neighborhood of every connected component of $G-X$, a clique. When $G$ is an incidence graph of the instance $\III$ and $X$ is a set of variables of $\III$, we also refer to $\torso_G(X)$ as $\torso_\III(X)$.
\end{definition}

  \begin{definition}\label{def:width}
      Let ${\cal F}$ be a class of CSP instances and $\III$ be a CSP instance. Then the \emph{backdoor-treewidth} of $\III$ with respect to $\cal F$, denoted $\sbw_{\cal F}(\III)$, is 
%defined as $\min \{\tw(\torso_\III(X)): X$ is a strong backdoor of $\III$ to ${\cal F} \}$. 
the smallest value of $\tw(\torso_\III(X))$ taken over all strong backdoors $X$ of $\III$ into $\cal F$.
If ${\cal F}=\CSP(\Gamma)$ for some constraint language $\Gamma$, then we call $\sbw_{\cal F}$ the backdoor-treewidth with respect to $\Gamma$.
\end{definition}

As an example, observe that in Figure~\ref{fig:example} the graph $\torso_G(X)$ is a path. Throughout this paper, we sometimes refer to backdoors of small treewidth simply as backdoors of small \emph{width}. Next, we show how backdoors of small treewidth can be used to solve CSP and $\sharpCSP$.

\lv{\begin{lemma}}
\sv{\begin{lemma}[$\star$]}
\label{lem:usingwidth}
 Let $\III$ be a CSP instance over domain $\cD$ and $X$ be a strong backdoor of $\III$ to the class ${\cal F}$. There is an algorithm that, given $\III$ and $X$, runs in time $\bigoh(\vert \cD\vert^{\tw(\torso(X))}\vert \III\vert^{\bigoh(1)})$ and correctly decides whether $\III$ is satisfiable or not. Furthermore, if $\cF$ is \sharpp-tractable and $X$ is a strong backdoor to $\cF$, then in the same time bound one can count the number of satisfying assignments of $\III$. \end{lemma}
 
 \lv{
\begin{proof} The algorithm is a standard dynamic programming procedure over a bounded treewidth graph and hence we only sketch it briefly. Let $G$ denote the incidence graph of $\III$ and let $H$ denote the graph $\torso(X)$ and let $(T,\cX)$ be a nice tree-decomposition of $H$ of width $\tw(H)$.
Now, for every $v\in T$, we define the instance $\III_{v}$ as the subinstance of $\III$ induced on the variables in $X_v$, the bags below it in $(T,\cX)$, and the constraints whose scope is completely contained in the union of $X_v$ and the bags below it. The key observation is that for any connected component of $G-X$, there is a vertex $v\in V(T)$ such that the bag $X_v$ contains the neighbors of this component. This is because these variables induce a complete graph in $\torso(X)$ and by the definition of tree-decompositions, every subgraph that is complete is contained in a bag of any tree-decomposition.

For each $v\in T$, we will define a function $\tau_v:\cD^{X_v}\to \{0,1\}$ which maps assignments of the variables in $X_v$ to 0 to 1. Let $\gamma:X_v\to \cD$ be an assignment to the variables in $X_v$. We define $\tau_v(\gamma)=1$ if there is a satisfying assignment for $\III_v$ that extends $\gamma$ and $\tau_v(\gamma)=0$ otherwise. Let $v^*$ denote the root of $T$. Clearly, the instance $\III$ is satisfiable if and only if there is a $\gamma:X_{v^*}\to \cD$ such that $\tau_{v^*}(\gamma)=1$. We now describe an algorithm to compute $\tau_v$ for every $v\in V(T)$.

It follows from the definition of $\tau$ that for every $u,v\in V(T)$ where $v$ denotes a \emph{forget} node and is a parent of $u$, the function $\tau_v$ can be computed from $\tau_u$. The same holds in the case of \emph{join} nodes. Therefore, it suffices to describe how to compute this function for \emph{leaf} nodes and \emph{introduce} nodes. Let $v$ be a leaf node and let $x$ be the unique variable in $X_v$. Consider the instance $\III_{v}$. We know by the definition of strong backdoors that the instance obtained from $\III_v$ by any instantiation of the variable $x$ is in the language $\Gamma$ which is assumed to be tractable. Hence we simply solve the instance resulting from $\III_v$ for every assignment to $x$. Now, let $v$ be an introduce node with child $u$. If there is a connected component of $G-X$ whose neighbors are in $X_v$ but not in $\III_u$, then we go over all instantiations of the variables in $X_v$ and solve the resulting tractable instance for each such component. Combining this with the function $\tau_u$ gives us the function $\tau_v$. Since one can also compute the number of satisfying assignments in a similar way, this concludes the proof of the lemma.
\end{proof}
}
\sv{
\begin{proof}[Sketch of Proof.]
The algorithm is a standard dynamic programming procedure over a bounded treewidth graph and hence we only sketch it briefly. Let $G$ denote the incidence graph of $\III$ and let $H$ denote the graph $\torso(X)$ and let $(T,\cX)$ be a tree-decomposition of $H$ of width $\tw(H)$.
Now, for every $v\in T$, we define the instance $\III_{v}$ as the subinstance of $\III$ induced on the variables in $X_v$, the bags below it in $(T,\cX)$, and the constraints whose scope is completely contained in the union of $X_v$ and the bags below it. The key observation is that for any connected component of $G-X$, there is a vertex $v\in V(T)$ such that the bag $X_v$ contains the neighbors of this component.

To solve CSP, for each $v\in T$ we define a function $\tau_v$ which maps assignments of the variables in $X_v$ to 0 to 1. Let $\gamma:X_v\to \cD$ be an assignment to the variables in $X_v$. We define $\tau_v(\gamma)=1$ if there is a satisfying assignment for $\III_v$ that extends $\gamma$ and $\tau_v(\gamma)=0$ otherwise. Let $v^*$ denote the root of $T$. Clearly, the instance $\III$ is satisfiable if and only if there is a $\gamma:X_{v^*}\to \cD$ such that $\tau_{v^*}(\gamma)=1$. At this point it suffices to describe how to dynamically compute the function $\tau_v$ for each node in the tree-decomposition; this step can be facilitated by the use of so-called nice tree-decompositions. The algorithm to solve $\sharpCSP$ is similar; there $\tau_v$ is extended by information about how many ways there are to extend an assignment to the variables in $X_v$ to a satisfying assignment for $\III_v$.
\end{proof}
}

As the width of a backdoor can be arbitrarily smaller than its size, the width provides a much better measure of how far away an instance is from a tractable base class. In particular, the width lower-bounds both the primal treewidth and the backdoor size. We formalize this below.

\lv{\begin{proposition}}
\sv{\begin{proposition}[$\star$]}
\label{lem:comparingwidth}
Let $\III$ be a CSP instance and $\cal F$ be a class of CSP instances. Let $q$ be the primal treewidth of $\III$ and $r$ be the minimum size of a strong backdoor to $\cal F$ in $\III$. Then $\sbw_{\cal F}(\III)\leq \min(q,r)$. 
%Similarly, if $\III$ is satisfiable and $r'$ is the minimum size of a weak backdoor to $\cal F$, then $\wbw_{\cal F}(\III)\leq \min(q,r')$.
\end{proposition}

\lv{
\begin{proof}
Observe that the graph $\torso(X)$ is a minor of the primal graph $G$ of $\III$. Indeed, to obtain $\torso(X)$ from $G$, it suffices to gradually contract all edges with an endpoint outside of $X$. Since minor operations can never increase the treewidth, it follows that $tw(\torso(X))\leq q$. Moreover, since the treewidth of a graph on $|X|$ vertices is upper-bounded by $|X|$, it follows that $tw(\torso(X))\leq r$ and $tw(\torso(X))\leq r'$, respectively.
\end{proof}
}

In order to prove Theorem~\ref{thm:main main theorem}, we give an {\FPT} algorithm for the problem of finding strong backdoors parameterized by their width (formalized below). 
We note that since we state our results in as general terms as possible, the dependence on $k$ is likely to be sub-optimal for specific languages and could be improved using properties specific to each language.

\defparproblem{{\sc Width Strong-$\CSP(\Gamma)$ Backdoor Detection}}{CSP instance $\III$, integer $k$.}{$k$}{Return a set $X$ of variables such that $X$ is a strong backdoor of $\III$ to $\CSP(\Gamma)$ of width at most $k$ or correctly conclude that no such set exists.
}

The main technical content of the article then lies in the proof of the following theorem.

\begin{theorem}
	\label{thm:find_backdoor}
{\sbddetection} is \emph{\FPT} for every finite $\Gamma$.
\end{theorem}

Before we proceed to the description of the algorithms, we state the following simple and obvious preprocessing routine (correctness is argued in the appended full version) which will allow us to infer certain structural information regarding interesting instances of this problem.%which will be executed on the given CSP instance.

\begin{redrule}\label{red:bounded_arity}
Given a CSP instance $\III$ and an integer $k$ as an instance of {\sbddetection}, if there is a constraint in $\III$ of arity at least $p+k+2$ where $p$ is the maximum arity of a relation in $\Gamma$, then return {\sc NO}.
\end{redrule}

\lv{We argue the correctness of this rule as follows. Suppose there is a constraint $C=((x_1,\dots,x_r),R)$ where $r\geq p+k+2$. Then, any strong backdoor set $X$ must contain at least $k+2$ variables in the scope of $C$. However, this implies that $\torso_{\psi(G)}(X)$ contains a clique on at least $k+2$ vertices, which in turn implies that $\sbw_\Gamma(\III)>k$. 
\lv{Moving forward, for any constraint language $\Gamma$ and integer $k\in \naturals$, we denote by $\rho(\Gamma,k)$ the integer $p+k+2$ where $p$ is the maximum arity of a relation in $\Gamma$. }
}

\section{The Finite State Lemma}
\label{sec:finitestate}
In this section, we prove that the problem {\sc Width Strong-$\CSP(\Gamma)$ Backdoor Detection} has finite state; this will allow us to construct a finite set of bounded-size representatives (Section~\ref{sec:rep}) which will play a crucial role in the proof of Theorem~\ref{thm:find_backdoor} (Section~\ref{sec:algo}).
Let $\Gamma$ be a finite constraint language; throughout the rest of the paper, we work with this fixed constraint language. We begin by defining a relation over the set of boundaried incidence graphs.

\begin{definition}
	Let $k,t\in \naturals$ and let $(G_1,Z_1)$ and $(G_2,Z_2)$ be $t$-boundaried incidence graphs of CSP instances $\III_1$ and $\III_2$ with boundaries $Z_1$ and $Z_2$ respectively. Then, we say that $(\III_1,Z_1)\sim_{t,k}(\III_2,Z_2)$ (or $(G_1,Z_1)\sim_{t,k}(G_2,Z_2)$) if for every $t$-boundaried CSP instance $\III_3$ with incidence graph $G_3$, the instance $\psi(G_1\oplus G_3)$ has a strong backdoor set of width at most $k$ into $\CSP(\Gamma)$ if and only if the instance $\psi(G_2\oplus G_3)$ has a strong backdoor set of width at most $k$ into $\CSP(\Gamma)$.
\end{definition}

It is clear that $\sim_{t,k}$ is an equivalence relation. Generally speaking, the high-level goal of this section is to prove that $\sim_{t,k}$ has finite index. This is achieved by introducing a second, more technical equivalence $\equiv_{t,h,\varepsilon}$ which captures all the information about how a $t$-boundaried incidence graph $(G,Z)$ contributes to the (non)-existence of a strong backdoor of small width after gluing. Observe that for a set $X$ which has vertices from `both' sides of a boundary the graph $\torso(X)$ may have edges crossing this boundary. Since we need to take this behaviour into account, proving this lemma is in fact much more involved than might be expected at first glance.

To define $\equiv_{t,h,\varepsilon}$, we will first need the notion of a \emph{configuration}, which can be thought of as one possible way a $t$-boundaried graph can interact via gluing; this is then tied to the notion of a \emph{realizable configuration}, which is a configuration that actually can occur in the graph $(G,Z)$. We let $(G_1,Z_1)\equiv_{t,h,\varepsilon}(G_1,Z_1)$ if and only if both boundaried graphs have the same set of realizable configurations.
Before we proceed to the technical definition of a configuration, we need one more bit of notation.
Since we will often be dealing with labeled minors, we fix a pair of symbols $\Box$ and $\Diamond$ and express all relevant label sets using these symbols.
\sv{Specifically, for $r,s\in \naturals$ we let $\mathfrak{L}(r,s)$ denote the set $2^{\{\Box_1,\dots, \Box_r\}\bigcup \{\Diamond_1,\dots, \Diamond_{s}\}}$.
}

\lv{
\begin{definition}
	Let $r,s\in \naturals$ and $T\subseteq \naturals$. We denote by $\mathfrak{L}(r,s)$ the set 
	$2^{\{\Box_1,\dots, \Box_r\}\bigcup \{\Diamond_1,\dots, \Diamond_{s}\}}$ and we denote by $\mathfrak{L}(T,s)$ the set $2^{\{\Box_i|i\in T\}\bigcup \{\Diamond_i| i\in [s]\}}$.
\end{definition}
}

\begin{definition}\label{def:config}
	Let $h,t\in \naturals$. A \textbf{$(t,h)$-configuration} is a tuple $(P,w,w',\cP,\cP',\gamma,\mathfrak{H})$, where: 
	\lv{
\begin{itemize}
\item $P$ is a subset of $[t]$,
\item $w,w' \in \mathbb{N}$ and $w'\leq (w+1)t$,
\item $\cP=\{Q_1,\dots, Q_r\}$ is a partition of $[t]\setminus P$, 
\item $\cP'\in 2^{P \choose 2}\times  2^{[w'] \choose 2}\times 2^{P\times [w']}$,
\item $\gamma:\cP\to 2^{P} \times 2^{[w']}$, 
\item $\mathfrak{H}$ is a collection of 
labeled graphs on at most $h$ vertices where the label set is $\mathfrak{L}(t,w')$
\end{itemize}
}
\sv{
\begin{enumerate*}[label={\textbf{(\arabic*)}}]
\item $P$ is a subset of $[t]$,
\item $w,w' \in \mathbb{N}$ and $w'\leq (w+1)t$,
\item $\cP=\{Q_1,\dots, Q_r\}$ is a partition of $[t]\setminus P$, 
\item $\cP'\in 2^{P \choose 2}\times  2^{[w'] \choose 2}\times 2^{P\times [w']}$,
\item $\gamma:\cP\to 2^{P} \times 2^{[w']}$, 
\item $\mathfrak{H}$ is a collection of 
labeled graphs on at most $h$ vertices where the label set is $\mathfrak{L}(t,w')$.
\end{enumerate*}
}

For a set $Q\in \cP$ with $\gamma(Q)=(J_1,J_2)$, we denote by $\gamma^i(Q)$ the set $J_i$ for each $i\in \{1,2\}$.
A $(t,h)$-configuration $(P,w,w',\cP,\cP',\gamma,\mathfrak{H})$ is called a {\bf $(t,h,\varepsilon)$-configuration} if $w\leq \varepsilon$ and we denote the set of such $(t,h)$-configurations by $\mathfrak{S}_{\sd}(t,h,\varepsilon)$. 
\end{definition}

Let us informally break down the intuition behind the above definition. $t$ corresponds to the size of the boundary of the associated $t$-boundaried incidence graph (as we will see in the next definition), and $h$ is an upper bound on the size of forbidden minors for our target treewidth. The $(t,h)$-configuration then captures the following information about interactions between a $t$-boundaried incidence graph $(G_1,Z_1)$ and a potential solution after gluing: 
%\vspace{-0.3cm}
\lv{

\begin{itemize}
\item $P$ represents the part of the boundary that intersects a backdoor of small width, 
\item $w'$ represents neighbors of the remainder of the boundary outside of $G_1$,
\item $w$ represents the target treewidth of the torso, 
\item $\cP$ represents how the part of the boundary outside of the strong backdoor will be partitioned into connected components, i.e., how it will `collapse' into the torso,
\item $\cP'$ represents all the new edges that will be created in the torso due to collapsing of parts outside of the torso,
\item $\gamma$ represents connections between connected components in the boundary outside of the strong backdoor and relevant variables in the strong backdoor, which is the second part of information needed to encode the collapse of these components into the torso,
\item $\mathfrak{H}$ represents `parts' of all minors of size at most $h$ present in the torso inside $G_1$.
\end{itemize}
}
\sv{(a) $P$ represents the part of the boundary that intersects a backdoor of small width, 
(b) $w'$ represents neighbors of the remainder of the boundary outside of $G_1$,
(c) $w$ represents the target treewidth of the torso, 
(d) $\cP$ represents how the part of the boundary outside of the strong backdoor will be partitioned into connected components, i.e., how it will `collapse' into the torso,
(d) $\cP'$ represents all the new edges that will be created in the torso due to collapsing of parts outside of the torso,
(e) $\gamma$ represents connections between connected components in the boundary outside of the strong backdoor and relevant variables in the strong backdoor, which is the second part of information needed to encode the collapse of these components into the torso,
(f) $\mathfrak{H}$ represents `parts' of all minors of size at most $h$ present in the torso inside $G_1$.
}
%\vspace{-0.3cm}
In order to formally capture the intuition outlined above, we define the result of `applying' a configuration on a $t$-boundaried incidence graph.

\begin{definition}\label{def:associate}
	Let $h,t\in \naturals$, $(G,Z)$ be the $t$-boundaried incidence graph of a $t$-boundaried CSP instance $\III$ and $\omega=(P,w,w',\cP,\cP',\gamma,\mathfrak{H})$ be a $(t,h)$-configuration. We associate with $G$ and $\omega$ an incidence graph $G^\omega$ which is defined as follows. We begin with the graph $G$, add $w'$ new variables $l^\omega_1,\dots, l^\omega_{w'}$, denoting the set comprising these vertices by $L_\omega$. For every $J\subseteq [w']$, we denote by $J(L_\omega)$ the set $\{l^\omega_i| i\in J\}$.
	 For each $Q\in {\cP}$, let $(J_1^Q,J_2^Q)=\gamma(Q)$ and add 
	$\vert Q\vert-1$ redundant binary constraints $C^Q_1,\dots, C^Q_{\vert Q\vert-1}$ (we have assumed that $\Gamma$ also contains a tautological relation of arity 2) and connect these with the variables in $Q(G,Z)$ to form a path which alternates between a vertex/variable in $Q(G,Z)$ and a vertex/variable in $\{C^Q_1,\dots, C^Q_{\vert Q\vert-1}\}$. Following this, for every variable $u$ in $J_1(G,Z)\cup J_2(L_\omega)$,  we add a redundant binary constraint $C_u$ and set $\var(C_u)$ as $u$ and an arbitrary variable in $Q(G,Z)$. This completes the definition of $G^\omega$. We also define the graph $\tilde G^\omega$ as the graph obtained from $G^\omega$ by doing the following.
	Let $\cP'=(X_1,X_2,X_3)$ where $X_1\subseteq {
P\choose 2}$, $X_2\subseteq {[w']\choose 2}$ and $X_3\subseteq P\times [w']$.
	 For every pair $(i,j)\in X_1$, we add the edge $(i(G,Z),j(G,Z))$. Similarly, for every pair $(i,j)\in X_2$, we add the edge $(l^\omega_i,l^\omega_j)$. Finally, for every pair $(i,j)\in X_3$, we add the edge $(i(G,Z),l^\omega_j)$.
	  This completes the description of $\tilde G^\omega$.
\end{definition}

The graph $G^\omega$ defined above can be seen as an enrichment of $G$ by (1) adding strong backdoor variables which will be affected by a collapse of the boundary into the torso ($l^\omega_1,\dots, l^\omega_{w'}$) and (2) enforcing the assumed partition of part of the boundary into connected components (as per $\cP$) and (3) adding connections of these components both into the rest of the boundary and vertices $l^\omega_{i}$ (as per $\gamma$). The graph $\tilde G^\omega$ is then an extension of $G^\omega$ by edges which will be created in the torso.
Note that while $G^\omega$ is an incidence graph, $\tilde G^\omega$ is not necessarily a bipartite graph.

With $\tilde G^\omega$ in hand, we can finally formally determine whether the information contained in a given configuration is of any relevance for the given graph. This is achieved via the notion of \emph{realizability}.
\begin{definition}

\label{def:realizability}	
Let $h,t\in \naturals$, $(G,Z)$ be the $t$-boundaried incidence graph corresponding to a $t$-boundaried CSP instance $\III$ and let $\omega = (P,w,w',\cP,\cP',\gamma,\mathfrak{H})$ be a $(t,h)$-configuration. We say that $\omega$ is a \textbf{realizable} configuration in $(G,Z)$ if, and only if, there is a subset $S^* \subseteq \cV(G)$ with the following properties:
\lv{
\begin{itemize}
\item $S^* \cap Z = P(G,Z)$
\item $tw(\torso_{\tilde G^\omega}(S^*\cup L_\omega))$ is at most $w$.
\item 
$\mathfrak{H} = \mathcal{M}_h(\torso_{\tilde G^\omega}(S^*\cup L_\omega),\Lambda_\omega)$, 
where  $\Lambda_\omega:S^*\cup L_\omega\to \mathfrak{L}(P,w')$ is defined as: 
 
 \begin{description} \item -- for all $v\in S^*\cap Z$, $\Lambda(v)=\{\Box_i\}$ where $v=i(G,Z)$, 
 	
\item
-- for all $v\in L_\omega$, $\Lambda(v)=\{\Diamond_i\}$ where $v=l^\omega_i$ and 
\item -- for every  vertex $v\in S^*\setminus Z$, $\Lambda(v)=\emptyset$. 
\end{description}
That is, $\mathfrak{H}$ is precisely the set of all labeled minors of $(\torso_{\tilde G^\omega}(S^*\cup L_\omega),\Lambda_\omega)$ with at most $h$ vertices.

\item $S^*$ is a strong backdoor of $\psi(G)$ into $\CSP(\Gamma)$.
\end{itemize}
If the above conditions hold, we say that $S^*$ \textbf{realizes} $\omega$ in $(G,Z)$. 
}
\sv{
\begin{enumerate*}[label={\textbf{(\arabic*)}}]
\item $S^* \cap Z = P(G,Z)$,
\item $tw(\torso_{\tilde G^\omega}(S^*\cup L_\omega))$ is at most $w$,
\item $\mathfrak{H}$ is precisely the set of all labeled minors of $(\torso_{\tilde G^\omega}(S^*\cup L_\omega),\Lambda_\omega)$ with at most $h$ vertices,
\item $S^*$ is a strong backdoor of $\psi(G)$ into $\CSP(\Gamma)$.
\end{enumerate*}
If the above conditions hold, we say that $S^*$ \textbf{realizes} $\omega$ in $(G,Z)$. 
}

\end{definition}

We let $\mathfrak{S}_{\sd}((G,Z),h,\varepsilon)$ denote the set of all realizable $(|Z|,h,\varepsilon)$-configurations in $(G,Z)$. 
We ignore the explicit reference to $Z$ in the notation if it is clear from the context.
We let $h^*(k)$
denote the upper bound on the size of forbidden minors for graphs of treewidth at most $k$ given in \cite{Lagergren98}. For technical reasons, we will be in fact concerned with minors of size slightly greater than $h^*(k)$, and hence for $t\in \mathbb{N}$ we set $h^*(k,t)=h^*(k)+t\cdot(k+1)$. 

We use $\Upsilon_{\sd}(t,h,\varepsilon)$ to denote a computable upper bound on the number of $(t,h,\varepsilon)$-configurations. Observe that setting $\Upsilon_{\sd}(t,h,\varepsilon)=2^t\cdot \varepsilon \cdot \varepsilon t \cdot t^t \cdot 2^{{t \choose 2}}\cdot 2^{  { {(\varepsilon +1) t}  \choose 2}} \cdot 2^{t^2(\varepsilon +1) } \cdot 2^{t^2(\varepsilon+1)} \cdot 2^{{h \choose 2}}h^{2^{(\varepsilon+1)t}}$ is sufficient. We now give the formal definition of the refined equivalence relation.

\begin{definition} Let $t,h\in \naturals$ and let $(\III_1,Z_1)$ and $(\III_2,Z_2)$ be $t$-boundaried CSP instances with $t$-boundaried incidence graphs $(G_1,Z_1)$ and $(G_2,Z_2)$ respectively. Then, $(\III_1,Z_1)\equiv_{t,h,\varepsilon}(\III_2,Z_2)$ (or $(G_1,Z_1)\equiv_{t,h,\varepsilon}(G_2,Z_2)$)  if $\mathfrak{S}_{\sd}((G_1,Z_1),h,\varepsilon)=\mathfrak{S}_{\sd}((G_2,Z_2),h,\varepsilon)$.
\end{definition}

\sv{
From these definitions, it is straightforward to verify that $\equiv_{t,h,\varepsilon}$ is indeed an equivalence and the number of equivalence classes induced by this relation over the set of all $t$-boundaried incidence graphs is at most 
$2^{\Upsilon_{\sd}(t,h,\varepsilon)}$. 
The main lemma of this section, Lemma~\ref{lem:replace_equiv}, then links $\equiv_{t,h,\varepsilon}$ to $\sim_{t,k}$, and in particular shows that the former is a refinement of the latter. We note that the more refined $\equiv_{t,h,\varepsilon}$ is used throughout the paper; it is not merely a tool for showing finite-stateness of $\sim_{t,k}$.
}

\lv{
Clearly, $\equiv_{t,h,\varepsilon}$ is an equivalence relation and the number of equivalence classes induced by this relation over the set of all $t$-boundaried incidence graphs is at most 
$2^{\Upsilon_{\sd}(t,h,\varepsilon)}$. 
We now formally prove that the equivalence relation $\equiv_{t,h^*(k,t),k}$ is a refinement of the equivalence relation $\sim_{t,k}$. That is, we prove that whenever 2 boundaried incidence graphs are in the same equivalence class of $\equiv_{t,h^*(k,t),k}$ then they are in the same equivalence class of $\sim_{t,k}$.
}

\lv{\begin{lemma}}
\sv{\begin{lemma}[$\star$]}
\label{lem:replace_equiv}
Let $k,t\in {\mathbb N}$ and let $(G_1,Z_1)$, $(G_2,Z_2)$ be two $t$-boundaried incidence graphs satisfying  $(G_1,Z_1)\equiv_{t,h^*(k,t),k}(G_2,Z_2)$. Then, $(G_1,Z_1)\sim_{t,k}(G_2,Z_2)$.
\end{lemma}

\lv{
\begin{proof} 
	In order to prove the lemma, we need to prove that 
 for any $t$-boundaried graph $(G_3,Z_3)$ the instance $\psi(G_1 \oplus G_3)$ has a strong backdoor set into $\CSP(\Gamma)$ of width at most $k$ if and only if the instance $\psi(G_2 \oplus G_3)$ has a strong backdoor set into $\CSP(\Gamma)$ of width at most $k$.
We first give a brief sketch of the proof strategy. We only present the proof of one direction of the statement as the proof for the converse can be obtained by simply switching $G_1$ and $G_2$ in the arguments. We begin by assuming the existence of a set $S_1$ which is a strong backdoor set of $\psi(G_1 \oplus G_3)$ into $\CSP(\Gamma)$ of width at most $k$. We then use the set $S_1$ to define a  $(t,h^*(k,t),k)$-configuration $\omega$ and argue that this is in fact a configuration realized by $S_1^*=S_1\cap V(G_1)$ in $(G_1,Z_1)$. We then use the premise that $(G_1,Z_1)\equiv_{t,h^*(k,t),k}(G_2,Z_2)$ to infer the existence of a set, say $S_2^*\subseteq \cV(G_2)$ such that $S_2^*$ realizes $\omega$ in $(G_2,Z_2)$. We then proceed to prove that the set obtained from $S_1$ by `cutting' $S_1^*$ and `pasting' $S_2^*$ is indeed a strong backdoor set of the required kind for the instance $\psi(G_2 \oplus G_3)$.

\noindent
\textbf{Phase I: Defining a $(t,h^*(k,t),k)$-configuration realized by $S_1^*$.}
Suppose that $\psi(G_1 \oplus G_3)$ contains a strong backdoor $S_1$ into $\CSP(\Gamma)$ such that $tw(\torso_{G_1 \oplus G_3}(S_1))$ $\leq k$. 
Unless specified otherwise, henceforth we use $\torso(S_1)$ to denote $\torso_{G_1\oplus G_3}(S_1)$. 
 Let $S_1^*=S_1\cap V(G_1)$ and let $L_1=\{l_1,\dots, l_z\}$ be the set of vertices in $S_1\setminus S_1^*$ which are adjacent to a component of $(G_1\oplus G_3)-S_1$ that intersects $Z_1$. Observe that $z\leq (k+1)\cdot t$, since otherwise there would be a component in $(G_1\oplus G_3)-S_1$ that is adjacent to more than $k+1$ variables in $S_1$, which in turn would result in a clique of size greater than $k+1$ in $\torso(S_1)$. 
 We now define a tuple  $\omega=(P,w,w',\cP,\cP',\gamma,\mathfrak{H})$ as follows (we will later show that $\omega$ is actually a configuration). 
 
 \begin{enumerate} 
\item Let $P\subseteq [t]$ such that $P(G_1,Z_1)=S_1\cap Z_1$.

\item  Let $w=k$.
\item   Let $w'=z$.
\item Let $\cP=\{Q_1,\dots, Q_r\}$ be the partition of $Z_1\setminus P$ such that for every $i\in [r]$, the variables in $Q_i$ are contained in the same connected component of $(G_1 \oplus G_3)-S_1$ and for every $i\neq j\in [r]$, the variables in $Q_i$ and $Q_j$ are in distinct connected components of $(G_1 \oplus G_3)-S_1$.  
\item  Let $X_1$ be the set of all pairs $(i,j)$ where $i,j\in P, i\neq j,$ and there is a component of $(G_1 \oplus G_3)-S_1$ which is adjacent to both $i(G_1,Z_1)$ and $j(G_1,Z_1)$ and disjoint from $V(G_1)$. Let $X_2$ be the set of all pairs $(i,j)$ where $i,j\in [z]$ and there is a component of $(G_1 \oplus G_3)-S_1$ which is adjacent to both $l_i$ and $l_j$ and disjoint from $V(G_1)$. Let $X_3$ be the set of all pairs $(i,j)$ where $i\in P$, $j\in [z]$ and there is a component of $(G_1 \oplus G_3)-S_1$ which is adjacent to both $i(G_1,Z_1)$ and $l_j$ and disjoint from $V(G_1)$.
 Finally, let $\cP'=(X_1,X_2,X_3)$. 
\item
Let $\gamma:\cP_1\to 2^{P} \times 2^{[z]}$ be the function defined as follows. For every $Q\in \cP_1$, let $J_1^Q$ denote the vertices of $P(G_1,Z_1)$ which are adjacent to the component of $(G_1\oplus G_3)-S_1$ that contains $Q$ and let $J_2^Q$ denote the vertices of $L_1$ which are adjacent to the component of $(G_1\oplus G_3)-S_1$ that contains $Q$. The function $\gamma$ is defined as $\gamma(Q)=(J_1^Q,J_2^Q)$ for every $Q\in \cP_1$. 
\item 
 We define a function 
 $\Lambda_1:(S_1^*\cup L_1)\to \mathfrak{L}(P,z)$  as follows. 
 \begin{itemize}
 \item  For every $v\in S_1^*\cap Z_1$, we set $\Lambda_1(v)=\{\Box_i\}$ where $v=i(G_1,Z_1)$;
 \item  for every $v\in L_1$, we set $\Lambda_1(v)=\{\Diamond_i\}$ where $v=l_i$; and 
 \item  for every other vertex $v\in (S_1^*\cup L_1)$, we set $\Lambda_1(v)=\emptyset$. 
 \end{itemize}
 Finally, we define $\mathfrak{H}$
    to be the set of all labeled minors of $(\torso_{G_1\oplus G_3}(S_1)[S_1^*\cup L_1],\Lambda_1)$ with at most $h^*(k,t)$ vertices. That is, $\mathfrak{H}=M_{h^*(k,t)}(\torso_{G_1\oplus G_3}(S_1)[S_1^*\cup L_1],\Lambda_1)$.

\end{enumerate}

We begin by showing that $\omega$ is indeed a $(t,h^*(k,t),k)$-configuration.

\begin{claim}
	$\omega$ is a $(t,h^*(k,t),k)$-configuration.
\end{claim}

\begin{proof}
 In order to prove this, we  only need to prove that $w'\leq (w+1)t$.
Since $\torso_{G_1\oplus G_3}(S_1)$ has treewidth at most $k$, it follows that any component of $(G_1 \oplus G_3)-S_1$ has at most $k+1$ neighbors in $S_1$ (otherwise $\torso_{G_1\oplus G_3}(S_1)$ would have a $(k+2)$-clique). Since $L_1$ is the set of vertices of $S_1$ which are neighbors of  $r\leq t$ components, it follows that $|L_1|\leq r(k+1)\leq t(k+1)$, implying that $w'=z\leq t(k+1)=t(w+1)$. This completes the proof of the claim.
\end{proof}

 Having proved that $\omega$ is a $(t,h^*(k,t),k)$-configuration, we now claim that $S_1^*$ in fact realizes $\omega$.

\begin{claim}
\label{claim:iso}
$S_1^*$ realizes $\omega$ in $(G_1,Z_1)$.	
\end{claim}

\begin{proof} In order to prove this, we need to argue that $S_1^*$ satisfies the properties in Definition \ref{def:realizability}. By the definition of $S_1^*$, it holds that $S_1^*\cap Z_1= P(G_1,Z_1)$. Hence the first property is satisfied.
 We now  argue that $S_1^*$ is a strong backdoor of $\psi(G_1)$ into $\CSP(\Gamma)$. Suppose that this is not the case and for some assignment $\tau:S_1^*\to \cD$, there is a constraint in $G_1$ whose associated relation after applying $\tau$ is not in $\Gamma$. However, since $S_1$ is a strong backdoor of $G_1\oplus G_3$ into $\CSP(\Gamma)$, it must be the case that this constraint contains in its scope a variable of $S_1\setminus S_1^*$. However, since $Z_1$ is comprised entirely of variables, no constraint in $G_1$ can contain in its scope a variable of $S_1\setminus S_1^*$, a contradiction. Hence, we conclude that $S_1^*$ is a strong backdoor of $\psi(G_1)$ into $\CSP(\Gamma)$, completing the argument for the fourth property in Definition \ref{def:realizability}.

In order to prove that the remaining two properties hold, we  show that $(\torso_{\tilde G^\omega}(S_1^*\cup L_\omega),\Lambda_\omega)$ has a \emph{label-preserving} isomorphism to the graph $(\torso_{G_1\oplus G_3}(S_1) [S_1^*\cup L_1],\Lambda_1)$. For ease of presentation, let $B_1$  denote the graph $\torso_{\tilde G^\omega}(S_1^*\cup L_\omega)$ and $B_2$ denote the graph $\torso_{G_1\oplus G_3}(S_1) [S_1^*\cup L_1]$.

 We now define a bijection $\phi:V(B_1)\to V(B_2)$. For every $v\in S_1^*$, we set $\phi(v)=v$. For every $l^\omega_i\in L_\omega$, we set $\phi(l^\omega_i)=l_i\in L_1$. We argue that $\phi$ is in fact a label-preserving isomorphism between $(B_1,\Lambda_\omega)$ and $(B_2,\Lambda_1)$. It is straightforward to verify that for any vertex $v\in V(B_1)$, $\Lambda_\omega(v)=\Lambda_1(\phi(v))$. Therefore, we only need to prove that $\phi$ is an isomorphism. We begin with the forward direction.

 ($\implies$) We show that for every edge $(u,v)\in E(B_1)$, $(\phi(u),\phi(v))$ is an edge in $B_2$. Consider an edge $(u,v)\in E(B_1)$. 

 \begin{description}
 	\item {\bf Case 1:} $u,v\in S_1^*$. 
 	By the definition of $B_1$, it must be the case that either there is a component $X$ of $\tilde G^\omega-(S_1^*\cup L_\omega)$ which is adjacent to both $u$ and $v$ or the pair $(u,v)\in X_1$ (see the description of $\cP'$ in the definition of $\omega$ ). In the former case, since $V(B_1)=S_1^*\cup L_\omega$, we consider the following two exhaustive subcases : $X\subseteq V(G_1)\setminus Z_1$ or $X\cap Z_1\neq \emptyset$.  Suppose that $X$ is disjoint from $Z_1$, that is, $X\subseteq V(G_1)\setminus Z_1$. Then, it must be the case that $N(X)\subseteq S_1^*$ and hence $X$ is also disjoint from $S_1$ and adjacent to $u$ and $v$ in $(G_1\oplus G_3) - S_1$, implying that $(u,v)$ is an edge in $B_2$. On the other hand, suppose that $X$ intersects $Z_1$. Then, there is a set $Q\in \cP$ such that $X\cap Z_1=Q$. By the definition of $\omega$, it follows that $Q$ is contained in a component $X'$ of $(G_1\oplus G_3) - S_1$. Now, the definition of $S_1^*$ and $\gamma$ implies that $X'$ is also adjacent to $u$ and $v$ in $(G_1\oplus G_3) - S_1$, implying that $(u,v)=(\phi(u),\phi(v))$ is an edge in $B_2$. 
 	
 	   In the latter case, that is when the pair $(u,v)\in X_1$, the description of the set $X_1$ implies that there is a component of $(G_1\oplus G_3) - S_1$ which is adjacent to both $u$ and $v$, implying that $(u,v)=(\phi(u),\phi(v))$ is an edge in $B_2$.
 	This completes the argument for the  first case.

 	\item {\bf Case 2:} $u,v\in L_\omega$. Let $u=l^\omega_i$ and $v=l^\omega_j$. By the definition of $B_1$ it follows that either  there is a set $Q\in \cP$ such that the set $\gamma^2(Q)$ contains $u$ and $v$ (recall that $\gamma^2(Q)$ denotes the set in the second co-ordinate of $\gamma(Q)$) or the pair $(i,j)\in X_2$. In the former case,  the definition of the function $\gamma$ implies that the component of $(G_1\oplus G_3) - S_1$ containing $Q$ is adjacent to the vertices $l_i,l_j\in L_1$. Hence, we infer that $(\phi(u),\phi(v))$ is an edge in $B_2$. In the latter case, the definition of the set $X_2$ implies that there is a component of $(G_1\oplus G_3) - S_1$ (not necessarily intersecting $V(G_1)$) that is adjacent to the vertices $l_i,l_j\in L_1$. Again, this implies that $(\phi(u),\phi(v))$ is an edge in $B_2$.  This completes the argument for the second case.
 	
 	\item {\bf Case 3:} $u\in S_1^*, v=l^\omega_i\in L_\omega$. In this case, it follows from the definition of $B_1$ that either there is a set $Q\in \cP_1$ such that the set $\gamma^2(Q)$ contains $v$ or the pair $(u,v)\in X_3$. Furthermore, in the former case, if $u\in Z_1$, then $u\in P(G_1,Z_1)$ and $\gamma^1(Q)$ contains $u$, which by the definition of $\omega$ implies that $u$ and $l_i$ are adjacent to the same component of $(G_1\oplus G_3)-S_1$, implying the edge $(u,l_i)=(\phi(u),\phi(l^\omega_i))$.
 	On the other hand, if $u\notin Z_1$, then the component containing $Q$ in $\tilde G^\omega-(S_1^*\cup L_\omega)$ is adjacent to $u$. Let this component be $X$. Then, the set $X\cap V(G_1)$ contains a path in $(G_1\oplus G_3) - S_1$ from $Q$ to a vertex which is adjacent to $u$. Note that this path is also present in $(G_1\oplus G_3)-S_1$. Also, the definition of $\gamma$ implies that since $v\in \gamma^2(Q)$, there is a path from $Q$ to a vertex that is adjacent to $l_i$ in the graph $(G_1\oplus G_3)-S_1$. Hence, we infer that there is a component of $(G_1\oplus G_3)-S_1$ that is adjacent to both $u$ and $l_i$ and hence $(u,l_i)=(\phi(u),\phi(v))$ is indeed an edge in $B_2$. 	
 	Finally, if the pair $(u,v)\in X_3$, then it follows from the definition of $X_3$ that there is a component of $(G_1\oplus G_3)-S_1$ which is adjacent to $u$ and $v$, implying that $(\phi(u),\phi(v))$ is an edge in $B_2$.
 	This completes the argument for the third and final case.
 \end{description}

($\impliedby$) We now argue the converse direction. That is, for every edge $(u,v)\in E(B_2)$, $(\phi^{-1}(u),\phi^{-1}(v))$ is an edge in $B_1$. By the definition of $B_2$, it follows that there is a component $X$ of $(G_1\oplus G_3) -S_1$ which is adjacent to $u$ and $v$. Since $V(B_2)=S_1^*\cup L_1$, we have the following exhaustive cases. 

\begin{description}
	\item {\bf Case 1:} $u,v\in S_1^*$. We have the following three subcases: $X\subseteq V(G_1)\setminus Z_1$, $X\cap Z_1\neq \emptyset$, or $X\cap V(G_1)=\emptyset$. If $X\subseteq V(G_1)\setminus Z_1$, then it follows that $X$ is also a connected component of $\tilde G^\omega-(S_1^*\cup L_\omega)$, implying that $(u,v)\in E(B_1)$. If $X\cap V(G_1)=\emptyset$, then the definition of $\cP'$ implies that $X_1$ contains the pair $(u,v)$, which in turn implies that $\tilde G^\omega$ and hence $B_1$ contains the edge $(u,v)$. Finally, if $X\cap Z_1\neq \emptyset$, then....
	
\item {\bf Case 2:} $u,v\in L_1$ where $u=l_i$ and $v=l_j$. Here, we have the following 2 subcases: $X\cap Z_1\neq \emptyset$ or $X\cap Z_1=\emptyset$. Since $Z_1$ intersects every connected set in $G_1\oplus G_3$ that contains vertices of $V(G_1)$ and $V(G_3)$, these 2 subcases are exhaustive. In the first subcase, suppose that $X\cap Z_1=Q$. Then, the definition of $\gamma$ implies that $\gamma^2(Q)$ contains $l^\omega_i$ and $l^\omega_j$. The definition of $\tilde G^\omega$ implies that $(\phi^{-1}(u),\phi^{-1}(v))$ is an edge in $B_1$. In the second subcase, the definition of $\cP'$ implies that the pair $(i,j)\in X_2$. Again, the definition of $\tilde G^\omega$ implies that $(\phi^{-1}(u),\phi^{-1}(v))$ is an edge in $B_1$, completing the argument for this case.

	\item {\bf Case 3:} $u\in S_1^*$, $v=l_i\in L_1$. Here, we have the following 2 subcases: $X\cap Z_1\neq \emptyset$ or $X\cap Z_1=\emptyset$. Again, since $Z_1$ intersects every connected set in $G_1\oplus G_3$ that contains vertices of $V(G_1)$ and $V(G_3)$, these 2 subcases are exhaustive. In the first subcase, suppose that $X\cap Z_1=Q$. Then, the definition of $\gamma$ implies that $\gamma^2(Q)$ contains $l^\omega_j$. If $u\in Z_1$ then $\gamma^1(Q)$ contains $u$, implying that $B_2$ contains the edge $(u,l^\omega_j)=(\phi^{-1}(u),\phi^{-1}(v))$. On the other hand, if $u\notin Z_1$, then the component of $\tilde G^\omega - (..)$ is adjacent to both $u$ and $l^\omega_j$, implying that $B_2$ contains the edge $(u,l^\omega_j)=(\phi^{-1}(u),\phi^{-1}(v))$ In the second subcase, it must be the case that $u\in Z_1$ and that the pair $(u,j)\in X_3$. Again, the definition of $\tilde G^\omega$ implies that $(\phi^{-1}(u),\phi^{-1}(v))$ is an edge in $B_1$, completing the argument for this final case.
\end{description}

Thus we have concluded that $\phi$ is an isomorphism.
Hence,  $tw(B_1)=tw(B_2)$ and since $B_2$ is a subgraph of $\torso_{G_1\oplus G_3}(S_1)$, it follows that $tw(B_1) \leq tw(\torso_{G_1\oplus G_3}(S_1))\leq k$.
Finally, since $\phi$ is also a label-preserving isomorphism between $(B_1,\Lambda_\omega)$ and $(B_2,\Lambda_1)$, we conclude that $\cM_{h^*(k,t)}(B_1,\Lambda_\omega)=\cM_{h^*(k,t)}(B_2,\Lambda_1)$ which is precisely $\mathfrak{H}$ by definition of $\omega$. This completes the proof of the claim that $\omega$ is realized by $S_1^*$ in $(G_1,Z_1)$.
\end{proof}

\noindent
\textbf{Phase II: Defining an equivalent solution in $\psi(G_2\oplus G_3)$.}

 Since $\omega\in \mathfrak{S}_{\sd}((G_1,Z_1),h^*(k,t),k)$, and the premise of the lemma guarantees that $\mathfrak{S}_{\sd}((G_1,Z_1),h^*(k,t),k)=\mathfrak{S}_{\sd}((G_2,Z_2),h^*(k,t),k)$, we conclude that $\omega$ is also realizable in $(G_2,Z_2)$. Let $S_2^*$ denote the subset of $V(G_2)$ that realizes $\omega$. We claim that $S_2=S_2^*\cup (S_1\setminus S_1^*)$ is in fact a strong backdoor of $\psi(G_2\oplus G_3)$ into $\CSP(\Gamma)$ and has width at most $k$. The rest of the proof of the lemma is dedicated to proving this statement.

    We begin by arguing that $S_2$ is indeed a strong backdoor of $G_2\oplus G_3$. Suppose that this is not the case and let $\tau:S_2\to \cD$ be an assignment to $S_2$ and $C$ be a constraint such that $C|_{\tau}\notin \Gamma$. If $C\in V(G_2)$ then this contradicts our assumption that $S_2^*$ is a strong backdoor of $G_2$. Therefore, it must be the case that $C\in V(G_3)$. Now, if $\var(C)$ is disjoint from $Z_3$, then we have a contradiction to our assumption that $S_1$ is a strong backdoor of $\psi(G_1\oplus G_3)$. Therefore, we conclude that $\var(C)$ intersects $Z_3$. However, since $S_1^*$ and $S_2^*$ realize  $\omega$, it must be the case that $\var(C)\cap Z_1=\var(C)\cap Z_2$. Therefore, if $S_1$ is a strong backdoor of $\psi(G_1\oplus G_3)$ to $\CSP(\Gamma)$,  then $S_2$ is a strong backdoor of $\psi(G_2\oplus G_3)$ into $\CSP(\Gamma)$. It remains to prove that $S_2$ is a strong backdoor of width at most $k$. For this, we need the following three claims.
    
The first claim states that the labeled minors ($h^*(k,\ell)$-folios) we expect to find in the torso of $G_2$ are actually there.
    \begin{claim}
    \label{claim:sameminors}
      Consider the graph $\torso_{G_2\oplus G_3}(S_2)[S_2^*\cup L_1]$ and the labeling $\Lambda_2:S_2^*\cup L_1 \to \mathfrak{L}(P,w')$ defined as follows. For every $v\in S_2^*\cap Z_2$,  $\Lambda_2(v)=\{\Box_i\}$ where $v=i(G_2,Z_2)$; for every $v\in L_1$, we set $\Lambda_2(v)=\{\Diamond_i\}$ where $v=l_i$; and for every other vertex $v\in (S_2^*\cup L_1)$, we set $\Lambda_2(v)=\emptyset$. Then, $\mathfrak{H}=\cM_{h^*(k,\ell)}(\torso_{G_2\oplus G_3}(S_2)[S_2^*\cup L_1],\Lambda_2)$.
    	
    \end{claim}

    \begin{proof}
    	By the definition of $S_2^*$, it holds that $\mathfrak{H}=\cM_{h^*(k,\ell)}(\torso_{\tilde G_2^\omega}(S_2)[S_2^*\cup L_\omega],\Lambda_\omega)$. Therefore, it suffices to prove that $\cM_{h^*(k,\ell)}(\torso_{G_2\oplus G_3}(S_2)[S_2^*\cup L_1],\Lambda_2)=\cM_{h^*(k,\ell)}(\torso_{\tilde G_2^\omega}(S_2)[S_2^*\cup L_\omega],\Lambda_\omega)$. In order to do so, we define a label-preserving isomorphism from $(\torso_{G_2\oplus G_3}(S_2)[S_2^*\cup L_1],\Lambda_2)$ to $(\torso_{\tilde G_2^\omega}(S_2)[S_2^*\cup L_\omega],\Lambda_\omega)$. The proof is identical to that of Claim \ref{claim:iso} and hence we do not repeat it.
    \end{proof}

 The second claim states that if some $\ell$-boundaried graphs $B_1$ and $B_2$ contain the same $(h^*(k)+\ell)$-folios, then their join with a boundaried graph $B_3$ must contain the same $h^*(k)$-folios. 
We will use one new piece of notation to make our exposition clearer. Given a minor $Q$ (constructed by a fixed sequence of deletions and contractions) in a graph $G$, we say that a vertex $v\in V(G)$ is a \emph{preimage} of a vertex $q\in Q$ iff if either $v=q$ or $v$ was contracted into a new vertex $v'$ which is a preimage of $q$.

    \begin{claim} \label{claim:composeminors}
    Let $\ell,h\in \naturals$ and $(B_1,K_1), (B_2,K_2)$ and $(B_3,K_3)$ be $\ell$-boundaried incidence graphs. For $i\in \{1,2\}$, let $H_i=B_i\oplus B_3$. For each $i\in \{1,2\}$, let $\Lambda_i:V(B_i)\to 2^{\{\bigtriangledown_1, \dots, \bigtriangledown_\ell \}}$ be defined as follows. For every $v\in V(B_i)\setminus K_i$, $\Lambda_i(v)=\emptyset$ and for every $v\in K_i$, $\Lambda_i(v)=\{\bigtriangledown_j\}$ where $v=\{j\}(B_i,K_i)$. If $\cM_{h+\ell}(B_1,\Lambda_1)=\cM_{h+\ell}(B_2,\Lambda_2)$ then $\cM_h(B_1\oplus B_3,\emptyset)=\cM_h(B_2\oplus B_3,\emptyset)$.
    \end{claim}

\begin{proof}
Consider a $h$-folio $Q$ in $\cM_h(B_1\oplus B_3,\emptyset)$. We intend to show that $Q$ is also present in $\cM_h(B_2\oplus B_3,\emptyset)$; the other direction is completely symmetric. For each $q\in Q$, let pre$_1(q)$ denote the set of preimages of $q$ in $B_1$ and $\text{pre}_3(q)$ the set of preimages of $q$ not in $B_1$. Observe that it may happen that pre$_1(q)$ is not a connected set, but only if pre$_1(q)$ intersects $K_1$. Let pre$_1(Q)=\SB X \SM \exists q\in Q \text{ s.t. } X\text{ is a connected component of pre}_1(q)\SE$. To capture the correspondence between pre$_1(Q)$ and $V(Q)$, we define the mapping map$:\text{pre}_1(Q)\rightarrow V(Q)$ where map$(X)=q\in Q$ iff $X$ is a connected component of a preimage of~$q$. 

Observe that since each set of preimages is disjoint from the others, $|$pre$_1(Q)|\leq |Q|+\ell$. So, let $Q_1$ be the $(h+\ell)$-folio obtained in $(B_1,\Lambda_1)$ by contracting each element in pre$_1(Q)$ into a single vertex and deleting all other vertices in $G_1$. Interestingly, observe that $Q_1$ need not necessarily be a subgraph of $Q$ (a vertex in $Q$ could be `split' into several vertices in $Q_1$). Since $\cM_{h+\ell}(B_1,\Lambda_1)=\cM_{h+\ell}(B_2,\Lambda_2)$, it follows that $Q_1$ also occurs as a $(h+\ell)$-folio in $(B_2,\Lambda_2)$. Let pre$_2(Q)$ denote the set of preimages of $Q_1$ in $B_2$. Note that there is a unique label-preserving one-to-one correspondence between the elements of pre$_2(Q)$  and those of pre$_1(Q)$, defined as follows: $f(X_2\in  \text{pre}_2(Q))= X_1\in  \text{pre}_1(Q)$ iff there exists $q\in Q_1$ such that both $X_1$ and $X_2$ are the preimages of $q$.

Now, let us consider the minor $Q'$ in $B_2\oplus B_3$ obtained by the following procedure. For each $q\in V(Q)$, we define $Y_q\subseteq V(B_2\oplus B_3)$ as follows: $Y_q=\text{pre}_3(q)\cup \SB x\in X_2 \SM X_2\in \text{pre}_2(Q) \text{ where map}(f(X_2))=q\SE$. Intuitively, $Y_q$ uses the correspondence between the preimages of $Q_1$ in $B_1$ and $B_2$ to replicate a preimage of $q$ in $B_2\oplus B_3$.
Now, for each $q\in V(Q)$ we let $q'$ be a vertex in $Q'$ obtained by contracting $Y_q$ into a single vertex. We claim that $Y_q$ is connected and hence that each such $q'$ is well-defined; indeed, each $X_2\in \text{pre}_2(Q)$ is itself connected by construction, and for each such $X_2$ there exists a corresponding $X_1$ with the same intersection with the boundary.

Finally, it remains to verify that for each vertex pair $a,b\in Q$ that is adjacent, the natural corresponding vertex pair $a', b'\in Q'$ is also adjacent. So, let $\bar a, \bar b$ be an adjacent pair of preimages of $a,b$, respectively, in $B_1$. If $\bar a, \bar b$ are adjacent due to an edge in $B_3$, then $\bar a, \bar b$ both occur in $B_3$ and hence they are also preimages of $a', b'$, from which the claim follows. On the other hand, if $\bar a, \bar b$ are adjacent due to an edge in $B_1$, then there exists at least one vertex $a_1\in Q_1$ (corresponding to $a$) and one vertex $b_1\in Q_1$ (corresponding to $b$) such that $a_1, b_1$ are adjacent in $Q_1$. But then the preimages of $a_1$ and of $b_1$ \textbf{in $B_2$} must also contain an adjacent pair, say $\bar a', \bar b'$. Finally, by construction of $Y_a$ and $Y_b$, we conclude that $\bar a'$ and $\bar b'$ must be preimages of $a', b'$, respectively. Thus $a', b'$ are indeed an adjacent pair in $Q'$. To conclude the proof, observe that, by deleting all vertices not contracted into $Q'$ and possibly some redundant edges, we have found a $h$-folio $Q$ which occurs in $B_2\oplus B_3$, and hence $Q\in \cM_{h}(B_2\oplus B_3,\emptyset)$.
\end{proof}

The next, final claim states that the torsos of the composed graphs can also be obtained by taking the respective parts of the torso and gluing these parts together. In other words, we show that if we have a part of a torso in each of the boundaried graph, then the whole torso can be obtained by simply gluing these parts along the correct boundary.

\begin{claim}
\label{claim:torsojoin}
Let 
\begin{itemize}
\item $B_1=\torso_{G_1\oplus G_3}(S_1)[S_1^*\cup L_1]$ and $K_1=P(G_1,Z_1)\cup L_1$;
\item $B_2= \torso_{G_2\oplus G_3}(S_2)[S_2^*\cup L_1]$ and $K_2=P(G_2,Z_2)\cup L_1$;
\item $B_3=\torso_{G_1\oplus G_3}(S_1)-(V(G_1)\setminus K_1)$ and $K_3=K_1= P(G_1,Z_1)\cup L_1$. 
\end{itemize}
Then, $\torso_{G_1\oplus G_3}(S_1)=(B_1,K_1)\oplus (B_3,K_3)$ and $ \torso_{G_2\oplus G_3}(S_2)= (B_2,K_2)\oplus (B_3,K_3)$.
\end{claim}

\begin{proof}
We prove $\torso_{G_1\oplus G_3}(S_1)=(B_1,K_1)\oplus (B_3,K_3)$, since the other claim is completely symmetric. For brevity, let $T=\torso_{G_1\oplus G_3}(S_1)$ and $B=(B_1,K_1)\oplus (B_3,K_3)$. First observe that $V(T)=V(B)$. Indeed, $B$ was obtained by partitioning the vertices of $T$ into $B_1$ and $B_3$ with the exception of $K_1$, which was copied into both $B_1$ and $B_3$. Then clearly gluing $B_1$ and $B_3$ together will merge the two distinct copies of each vertex in $K_1$ into a single vertex, hence resulting in the same vertex set as $V(T)$. For the same reason, any edge that is present in $B$ must also occur in $T$ ($B$ was obtained by joining two boundaried induced subgraphs of $T$).

So, what remains to show is that any edge $e=ac$ in $T$ also occurs in $B$. Clearly, if $a,c\in V(B_1)$ then $e\in E(B_1)$ and in particular $e\in E(B)$. For the same reason, if $a,c\in V(B_3)$ then $e\in E(B)$ as well. So, consider $a\in V(B_1)\setminus V(B_3)$ and $c\in V(B_3)\setminus V(B_1)$ and assume for a contradiction that $e=ac\in E(B)$. In particular, since $K_1=V(B_1)\cap V(B_3)$, it follows that neither $a$ nor $c$ may occur in $K_1$. By the construction of a torso, this implies that there exists an $a$-$c$ path in $G_1\oplus G_3$ which does not intersect $K_1$. Since $Z_1$ is the boundary of $G_1$, the set $Z_1$ must also be a separator between $V(G_1)\setminus Z_1$ and $V(G_3)\setminus Z_3$ in $G_1\oplus G_3$, and in particular $P$ must intersect $Z_1$. Let $z$ be the last vertex of $Z_1$ in $P$, and in particular $z\in Z_1\setminus K_1$. Since $P$ is a path which ends in $c$ and $Z_1$ is a separator, the vertex $z_2$ immediately following after $z_1$ on $P$ must lie in $V(G_3)\setminus Z_3$. Once again, by our assumptions about $P$ we have $z_2\not \in K_1$ and in particular $z_2\not \in L_1$. But then, by the construction of $L_1$, $z_2\not \in S_1$ and hence $z_2\neq c$. So, let $D$ be the connected component of $z_2$ of $(G_1\oplus G_3)-S_1$ and observe that $D$ contains $z_1$ and hence intersects $Z_1$. Since $P$ ends in $c$, which is a vertex in $S_1\setminus S_1^*$, there must exist a vertex $d'$ which is the first vertex on $P$ in $S_1\setminus S_1^*$ after $z_2$. But then $d'$ is adjacent to $D$, and by the constructiion of $L_1$ it follows that $d'$ must necessarily be in $L_1\subseteq K_1$. This contradicts our assumption about $P$ not intersecting $K_1$, and we conclude that $B$ cannot contain any edge with one endpoint in each of $V(B_1)\setminus V(B_3)$, $V(B_3)\setminus V(B_1)$. 

Since we have shown that $T$ and $B$ have the same vertex set, any edge in $B$ occurs in $T$, and also that any edge in $T$ occurs in $B$, the claim holds.
\end{proof}

To complete the proof of Lemma~\ref{lem:replace_equiv}, consider for a contradiction that $\tw(\torso_{G_2\oplus G_3}(S_2)) >k$ and let $\ell=t\cdot (k+1)\geq |K_i|$ for $i\in [3]$. Then $\torso_{G_2\oplus G_3}(S_2)$ contains a forbidden minor, say $Q$, for treewidth $k$, and such a forbidden minor has size at most $h^*(k)$. By Claim~\ref{claim:torsojoin}, $\torso_{G_2\oplus G_3}(S_2)= (B_2,K_2)\oplus (B_3,K_3)$. Furthermore, by Claim~\ref{claim:sameminors}, $(B_2,K_2)$ contains the same $h^*(k,\ell)$-folios as $(B_1,K_1)$. But then it follows from Claim~\ref{claim:composeminors} that $Q$ is also a minor in $\torso_{G_1\oplus G_3}(S_1)$, contradicting our assumption that $\tw(\torso_{G_1\oplus G_3}(S_1))\leq k$.    

            Hence, we conclude that $S_2$ has width at most $k$, thus completing the proof of the lemma.
    \end{proof}

%Given the above lemma, Lemma \ref{lem:finite_state_lemma} follows by setting $\mathfrak{g}(t,k)=2^{\Upsilon(t,h^*(k,t),k)}$. 
Before we move ahead to the next section, we  state the following lemma, the proof of which 
is identical to the `cut' and `paste' argument in the previous lemma and hence we do not repeat it.
\begin{lemma}\label{lem:local_cut_paste}
	Let $k,t\in \naturals$ and let $(G,Z)$ be a $t$-boundaried incidence graph. Let $(G',Z')$ be a $t$-boundaried incidence graph. Let $S$ be a strong backdoor of $\psi(G\oplus G')$ into $\CSP(\Gamma)$ of width at most $k$ and let $X=S\cap V(G)$. Let $\omega$ be a $(t,h^*(k,t),k)$-configuration realised in $(G,Z)$ by $S^*$ where $\omega$ is defined as in the proof of the previous lemma. Then, for any set $X'\subseteq V(G)$ that realises $\omega$, the set $(S\setminus X)\cup X'$ is a strong backdoor of $\psi(G\oplus G')$ into $\CSP(\Gamma)$ of width at most $k$.
\end{lemma}
}
%%%%%%%

\section{Computing a Bound on the Size of a Minimal Representative of $\sim_{t,k}$}
\label{sec:rep}

In this section, we define a function $\alpha_{\sd}$ such that for every $t,k\in \naturals$, every equivalence class of $\sim_{t,k}$ contains a boundaried incidence graph whose size is bounded by $\alpha_{\sd}(t,k)$. In order to do so, we use the fact the relation $\equiv_{t,h^*(k,t),k}$ refines $\sim_{t,k}$. The following is a brief sketch of the proof strategy.
%In this section, we define a function $\alpha_{\sd}$ such that for every $t,k\in \naturals$, every equivalence class of $\equiv_{t,h^*(k,t),k}$ contains a boundaried incidence graph whose size is bounded by $\alpha_{\sd}(t,k)$. The following is a brief sketch of the proof strategy. It has two steps.
\begin{itemize}
	\item  In the first step \lv{(Lemma \ref{lem:reduction_strong})}, we show that for any $t$-boundaried incidence graph $(G,Z)$ whose \emph{treewidth} is bounded as a function of $t$ and $k$ and size exceeds a certain bound also depending only on $t$ and $k$, there is a strictly smaller $t$-boundaried graph $(G',Z')$ such that $(G,Z)\equiv_{t,h^*(k,t),k}(G',Z')$. This in turn implies that for any $t$-boundaried incidence graph $(G,Z)$ whose \emph{treewidth} is bounded by a function of $t$ and $k$ there is a $t$-boundaried graph $(G',Z')$ such that $(G,Z)\equiv_{t,h^*(k,t),k}(G',Z')$ and the size of $G'$ is bounded by a function of $t$ and $k$.  
	\item In the second step \lv{(Lemma \ref{lem:strong_rep_bound})}, we show that for any $t$-boundaried incidence graph $(G,Z)$, there is a $t$-boundaried incidence graph $(G',Z')$ such that $G'$ has treewidth bounded by a function of $k$ and $t$ and $(G,Z)\sim_{t,k}(G',Z')$.  \sv{Combining these two steps, we obtain the following lemma. }
	 \end{itemize}

\lv{
For the following lemma, let $\mathfrak{b}$ be the function bounding the primal treewidth based on the arity and incidence treewidth specified in Proposition~\ref{prop:treewidth_torso_bound}, i.e., $\mathfrak{t}(t,\rho)=\rho(t+1)-1$.

\begin{lemma}\label{lem:treewidth_torso_bound}
%Let $\mathfrak{t}:\naturals^2 \to \naturals$ be defined as $\mathfrak{t}()=$.
%There is a function $\mathfrak{t}:\naturals^2 \to \naturals$ such that for every
 Let $\rho\in \naturals$  and  $G$  be an incidence graph where every constraint has arity at most $\rho$. Then, for every $X\subseteq \cV(G)$,  $\tw(\torso_G(X))\leq \mathfrak{t}(tw(G),\rho)$.
\end{lemma}

\begin{proof} By Proposition~\ref{prop:treewidth_torso_bound}, $\mathfrak{t}(tw(G),\rho)$ is an upper bound on the treewidth of the primal graph (call it $H$) of $\psi(G)$. However, observe that for any $X\subseteq \cV(G)$, the graph $\torso_G(X)$ is a minor of $H$ and hence $tw(\torso_G(X))\leq \mathfrak{t}(tw(G),\rho)$. This completes the proof of the lemma.
\end{proof}
}
\lv{
\lv{\begin{lemma}}
\sv{\begin{lemma}[$\star$]}
\label{lem:reduction_strong}
 There is a function $\xi_{\sd}:\naturals^3\to \naturals$ such that for all $t,k, \ell\in \naturals$, for every $t$-boundaried incidence graph $(G,Z)$ with treewidth at most $\ell$ and size at  least $\xi_{\sd}(k,t,\ell)$, there is  a strictly smaller $t$-boundaried incidence graph $(G',Z')$ such that $(G,Z)\equiv_{t,h^*(k,t),k}(H,J)$. 
\end{lemma}
}

\lv{
\begin{proof}
%Recall that the treewidth of $\torso_G(Y)$ is at most $a$ where $a=a(k,\ell)$ depends only on the arity of constraints in $G$  and $\ell$ and hence on $k$, the fixed language $\Gamma$ and $\ell$. 
Let $H$ denote the primal graph of $\psi(G)$. Since we are only interested in arguing the existence of a boundaried incidence graph $(G',Z')$ such that
$(G,Z)\equiv_{t,h^*(k,t),k}(G',Z')$, we may assume that the constraints in $\psi(G)$ have arity at most $\rho=\rho(\Gamma,k)$  (see Reduction Rule \ref{red:bounded_arity}). 
 Applying Proposition \ref{prop:treewidth_torso_bound}, we conclude that $tw(H)\leq \mathfrak{t}(\ell,\rho)$.

Let $(T,\cX)$ be a nice tree-decomposition of $H$ of width $tw(H)$ and let $(T',\cX')$ denote the tree-decomposition resulting from $(T,\cX)$ by adding $Z$ to every bag. Observe that the width of the decomposition $(T',\cX'=\{X_v|v\in V(T')\})$, denoted by $\ell'$, is $tw(H)+t\leq \mathfrak{t}(\ell,\rho)$.  Since $(T,\cX)$ is rooted by definition,  so is $(T',\cX')$. For technical reasons, we also create a bag containing only the vertices of $Z$, add it to the tree-decomposition $(T',\cX')$ by making it adjacent to the root and make this new bag the new root.

%Let $d=a+\ell+1$. Then, every bag in this tree-decomposition has size at most $d$.  
Now, for every $v\in T'$, we define the incidence graph $G_{X_v}$ as the subgraph of $G$ induced on the variables in $X_v$ and the bags below it in $T'$, and the constraints whose scope is completely contained in the union of $X_v$ and the bags below it. 
We now define the notion of a pair of equivalent bags in $T'$. For $u,v\in V(T')$, we say that $X_u$ and $X_v$ are \emph{equivalent} if they have the same size and the boundaried incidence graphs $G_{X_u}$ and $G_{X_v}$ with boundaries $X_u$ and $X_v$ (annotated by some $\lambda_u:X_u\to [\vert X_u\vert]$ and $\lambda_v:X_v\to [\vert X_v\vert]$) are equivalent with respect to $(|X_u|,h^*(k,|X_u|),k)$-configurations. We argue that if $\cV(G)$ is large enough, then $(T',\cX')$ contains a pair of equivalent bags. We first prove the following claim.

	\begin{claim}
\label{clm:threshold}
There is a constant $c(k,t,\ell)$ such that if $|\cV(G)| > c(k,t,\ell)$, then $(T',\cX')$ contains two equivalent bags $X_u$ and $X_v$, such that $u$ is an ancestor of $v$. 
\end{claim}

\begin{proof} Note that $T$ has at least $c(k,t,\ell)$ vertices corresponding to introduce nodes. Further, since $T$ is a binary tree, and a binary tree on $2^{n} - 1$ vertices has at least one root-to-leaf path of length at least $n$, we have that $T$ admits a root-to-leaf path, say $P$, of length at least $\lfloor \log c(k,t,\ell)\rfloor$. Now, since the number of subsets of the set of all $(d,h^*(k,d),k)$-configurations of $d$-boundaried graphs is bounded by $2^{\Upsilon_{\sd}(d,h^*(k,d),k)}$, we conclude that if $\lfloor \log c(k,t,\ell)\rfloor>2^{\Upsilon_{\sd}(d,h^*(k,d),k)}$, then there is indeed a pair of equivalent bags (in $T$ and hence in $T'$) with one being an ancestor of the other. Therefore, setting $c=2^{2^{\Upsilon_{\sd}(d,h^*(k,d),k)}+1}$ concludes the proof of the claim. 
\end{proof}

Now, let $u,v\in V(T')$ be such that 	$X_u$ and $X_v$ are equivalent bags and $u$ is an ancestor of $v$ in $T'$. We now argue that $(G,Z)\equiv_{t,h^*(k,t),k}(G',Z')$ where $(G',Z')$ is defined as the boundaried graph obtained from $(G,Z)$ by replacing the graph $G_{X_u}$ with $G_{X_v}$. Once we argue this, the lemma follows by choosing $\xi_{\sd}(k,t,\ell)$ to be the same as $c(k,t,\ell)$.

\begin{claim}
	Let $(G',Z')$ be defined as the boundaried graph obtained from $(G,Z)$ by replacing the boundaried graph $(G_{X_u},X_u)$ with $(G_{X_v},X_v)$. Then, $(G,Z)\equiv_{t,h^*(k,t),k}(G',Z')$.
\end{claim}

\begin{proof}
%The proof follows a similar argument as the proof of Lemma~\ref{lem:replace_equiv}, and so we merely provide a high-level summary.
Assume for a contradiction that there exists a $(t,h^*(k,t),k)$-configuration $\alpha$ such that, w.l.o.g., $\alpha\in \mathfrak{S}(G,Z)\setminus \mathfrak{S}(G',Z')$, and let $\alpha$ be realized in $(G,Z)$ by $S^*$. Let $S^*_u=S^*\cap V(G_{X_u})$, and let $\alpha_{uv}$ be the $(t,h^*(k,t),k)$-configuration realized by $S^*_u$ in $G_{X_u}$. Since $\mathfrak{S}(G_{X_u},X_u)=\mathfrak{S}(G_{X_v},X_v)$, we have that $\alpha_{uv}\in \mathfrak{S}(G_{X_v},X_v)$ and in particular there exists a variable-subset $S^*_v\subseteq \cV(G_{X_v})$ which realizes $\alpha_{uv}$ in $G_{X_v}$. It remains to argue that the set $S^*_2=(S^*\setminus S^*_u)\cup S^*_v$ in fact realizes the $(t,h^*(k,t),k)$-configuration $\alpha$ in $(G',Z')$, which follows by an analogous chain of arguments as the proof of Lemma~\ref{lem:replace_equiv}. This then yields a contradiction with the assumption that $\alpha\not\in \mathfrak{S}(G',Z')$.
%Then, by (the proof of) Lemma~\ref{lem:replace_equiv} there exists $(G'',Z'')$ together with a variable-subset $Y''\subseteq \cV(G'')$ which witnesses the realizability of $\alpha$. 
\end{proof}

This completes the proof of the lemma.
\end{proof}

}

\lv{
\lv{\begin{lemma}}
\sv{\begin{lemma}[$\star$]}
	\label{lem:strong_rep_bound}
	There is a function $\iota_{\sd}:\naturals^2\to \naturals$ such that for every $k,t\in {\mathbb N}$ and  $t$-boundaried incidence graph $(G,Z)$ there is a $t$-boundaried incidence graph $(G',Z')$ of treewidth at most $ \iota_{\sd}(k,t)$ such that $(G,Z)\sim_{t,k}(G',Z')$. 
\end{lemma}
}

\lv{
\begin{proof}
 For every  $(t,h^*(k,t),k)$-configuration $\omega\in \mathfrak{S}_{\sd}((G,Z),h^*(k,t),k)$, we denote by $S^*_\omega$ an arbitrary subset of $V(G)$ realizing $\omega$.
% For each $S_\omega$, we define $S^*_\omega$ as $S_\omega \cap \cV(G_1)$. 
 We now define a set $Y\subseteq V(G)$ as $Y=Z\cup \bigcup_{\omega\in \mathfrak{S}_{\sd}((G,Z),h^*(k,|Z|),k) }S^*_\omega$. Before we proceed we need the following claim about the structure of $Y$.  
 
 \begin{claim}
 $tw(\torso_{G}(Y))\leq \Upsilon_{\sd}(t,h^*(k,t),k)\cdot (k+1)+t$.	 
% Furthermore, for every $t$-boundaried $(H,J)$, $G\oplus H$ has a strong backdoor of width at most $k$ if and only if $G\oplus H$ has such a strong backdoor disjoint from $\cV(G)\setminus Y$.
 \end{claim}

\begin{proof} 
%The second part of the claim follows from the definition of $Y$. Therefore, we only need to prove the treewidth bound on $\torso(Y)$. 
Let $\omega_1,\dots,\omega_r$ be the realizable configurations in $\mathfrak{S}_{\sd}((G,Z),h^*(k,|Z|),k)$. We define $S^*_{i}$ as $\bigcup_{q\in [i]}S^*_{\omega_i}$. 
We will show that for every $i\in [r]$, $tw(\torso_G(S^*_{i}))\leq i\cdot (k+1) +t$ and there is a tree-decomposition $\cT_i$  of $\torso_G(S^*_{i})$ of width at most $i\cdot (k+1) +t$ such that the neighborhood of every connected component of $G-S^*_i$ is contained in some bag of this tree-decomposition.
% and there is a tree-decomposition $\cT_i$ of this width where every bag contains the boundary $Z$ .

  The proof is by induction on $i$. Consider the case when $i=1$.  Since $S^*_{\omega_1}$ realizes $\omega_1$, we have that $tw(\torso_G(S^*_{1}))\leq k$. We then add the vertices of $Z$ to all bags of an arbitrary tree-decomposition of $\torso_G(S^*_{1}))$ of this width to get a tree-decomposition of width at most $k+t$, which we call $\cT_1$. Observe that since the neighborhood of every connected component of $G-S^*_1$ is now a clique in  $\torso_G(S^*_{1}))$, it follows that the neighborhood of any connected component of $G-S^*_1$ is contained in some bag of $\cT_1$. We now consider the case when $i>1$.
  
    By the induction hypothesis, we have that $\tw(\torso_G(S^*_{i-1}))\leq (i-1)\cdot (k+1)+t$. Furthermore, there is a tree-decompsition $\cT_{i-1}$ with this width such that if 
   $C_1,\dots,C_s$ are the connected components of $G- S^*_{i-1}$ then the neighborhood of each $C_j$ is contained in some bag of the tree-decomposition $\cT_{i-1}$.

 For each $j\in [s]$, let $D_j=S^*_{\omega_i}\cap C_j$.  We know that there is a tree-decomposition $\cT'_j$ of $\torso_{G[C_j]}(D_j)$ of width at most $k$. Further, we have that every component of $G[C_j]-D_j$ has at most $k+1$ neighbors in $D_j$ and at most $(i-1)\cdot (k+1) +t$ neighbors in $S^*_{i-1}$. We now redefine $\cT'_j$ as follows. We add the vertices in $N(C_j)$ to every bag of $\cT'_j$.  We then take the tree decomposition $\cT_{i-1}$ and for each $j\in [s]$, we make an arbitrary bag of $\cT'_j$ adjacent to an arbitrary bag of $\cT$ which contains $N(C_j)$. Observe that what results is indeed a tree-decomposition of $\torso_G(S^*_i)$ and we call this tree-decomposition $\cT_i$. It follows from  definition that the width of $\cT_i$  exceeds the width of $\cT_{i-1}$ by at most $k+1$. Hence, the width of $\cT_i$ is at most $i\cdot (k+1)+t$. Furthermore, observe that for every $j$, the neighborhood of every connected component of $G[C_j]-D_j$ \emph{within} $C_j$ is contained in some bag of $\cT'_j$. By the construction of $\cT_i$, we can conclude that any connected component of $G-S^*_i$ is in fact a connected component of $G[C_j]-D_j$ for some $j$, and furthermore the neighborhood of such a component is contained in some bag of $\cT_i$.  This completes the proof of the claim.
\end{proof}

Having proved the claim, we now return to the proof of the lemma.
Since $tw(\torso_{G}(Y))\leq \Upsilon_{\sd}(t,h^*(k,t),k)\cdot k+t$, we conclude that every connected component of $G\setminus Y$ has at most $\beta$ neighbors in $Y$, where $\beta=\Upsilon_{\sd}(t,h^*(k,t),k)\cdot k+t+1$. We will use this fact to replace large components outside of $Y$ with small ones while preserving equivalence (informally speaking, these are constructed by keeping sufficiently many constraints to preserve interactions with $Y$, making redundant copies of variables and constraints to prevent a backdoor from using the component, and adding complete connections between the new variables).

%We combine Lemma \ref{lem:local_cut_paste} with the fact that for every realizable $(t,h^*(k,t),k)$-configuration in $(G,Z)$ , there is a subset of $Y$ that realizes it, to conclude that for 
%every $t$-boundaried graph $(H,J)$, there is a strong backdoor set of $(G\oplus H)$ of width at most $k$ if and only if there is one which avoids the variables in $\cV(G)\setminus Y$. 

So, let $C_1,\dots, C_s$ be the connected components of $G-Y$. 
%Let $\gamma=|\Gamma|^{\cD^{2\beta}}$. 
For each $i\in [s]$, we define a function $\sigma_i$ that maps each constraint $c\in C_i$ to a $2^{|N(C_i)|} \times \cD^{|N(C_i)|}$ matrix $M_{c,i}$ with elements from $1,\dots,\kappa$, where $\kappa=|\cR^{\rho}_\cD|$, the number of possible relations of arity at most $\rho$ over the domain $\cD$. The rows of the matrix correspond to subsets of $N(C_i)$ and the columns correspond to assignments to the variables in $N(C_i)$. For a set $P\subseteq N(C_i)$ and assignment $\tau:N(C_i)\to \cD$, the corresponding cell of $M_{c,i}$ has the value $q\in [|\Gamma|]$ if reducing the constraint $c=(S,R)$ with the assignment $\tau|_P$ results in a constraint $c'=(S',R')$ where $R'$ is the $q^{th}$ relation in $\cR^{\rho}_\cD$. We may assume that the relations in $\cR^{\rho}_\cD$ are arbitrarily ordered. Note that we do not claim that the relation $R'$ is in our language $\gamma$.

Observe that the range of the function $\sigma_i$ for any $i\in [s]$ is bounded by $\eta=\kappa^{\cD^{2\beta}}$ where $\kappa$ (in our context) is bounded by $2^{\cD^{\rho(\Gamma,k)}}$. Now, for each $i\in [s]$, we pick a set $\CCC_i$ of at most  $\eta$ constraints as follows. If the number of constraints in $C_i$ is at most $\eta$, then we add all constraints in $C_i$ to $\CCC_i$. Otherwise, for every image of the function $\sigma_i$, we pick an arbitrary pre-image and add this constraint to $\CCC_i$. Observe that for every $i\in [s]$, $|\CCC_i|\leq \eta$. 
We now define the set $\cV_i$ as the set of all variables disjoint from $N(C_i)$ which occur in the scope of a constraint in $\CCC_i$. For every variable $v\in \cV_i$, we make $k+2$ copies denoted $\{v_1,\dots, v_{k+2}\}$, and for every constraint $c$ whose scope includes $v$ we make $k+2$ copies of this constraint $c_1,\dots, c_{k+2}$ with $v_i$ belonging to the scope of $c_i$. In order to keep the presentation simple, we continue to use $\CCC_i$ to refer to the larger set of constraints introduced by this operation. We now add a set $\hat{\CCC}_i^v$ of ${k+2 \choose 2}$ new constraints, each of which is a redundant tautological binary constraint with a distinct pair of copies of $v$ as its scope. We define the set $\hat \CCC_i$ as $\bigcup_{v\in \cV_i}\hat \CCC_i^v$. We use $\hat \cV_i$ to denote the set containing all $k+2$ copies of all vertices in $\cV_i$.

 We then introduce a set $\CCC_i'$ of ${{|(\hat \cV_i\cup N(C_i))|} \choose 2}$ new constraints, each of which is a redundant tautological binary constraint with a distinct pair of variables in $\hat \cV_i\cup N(C_i)$ as its scope.
We now define $G'=G[Y\bigcup_{i\in [s]} (\CCC_i\cup \hat \CCC_i\cup  \CCC'_i\cup \hat \cV_i \cup N(C_i))]$ and $Z'=Z$ and claim that $(G',Z')\sim_{t,k} (G,Z)$. Note that since $Z\subseteq Y$ by definition, $Z'\subseteq V(G')$ and hence $(G',Z')$ is indeed a $t$-boundaried graph.  In order to complete the proof of the lemma, we need the following claims.

\begin{claim}\label{clm:first} A set $S\subseteq Y$ is a strong backdoor set of $G$ into $\CSP(\Gamma)$ if and only if it is also a strong backdoor set of $G'$ into $\CSP(\Gamma)$.
	
\end{claim}

\begin{proof}
 For the forward direction, since all the non-redundant constraints in $G'$ are already present in $G$, it follows that  if $S$ is a strong backdoor set of $G$ into $\CSP(\Gamma)$ then $S$ is  also a strong backdoor set of $G'$ into $\CSP(\Gamma)$. For the converse, suppose that $S$ is a strong backdoor set of $G'$ into $\CSP(\Gamma)$ and is \emph{not} a strong backdoor set of $G$ into $\CSP(\Gamma)$. Let $c\in \CCC(G)$ be a constraint and $\tau:S\to \cD$ be such that the relation of $c|_\tau$ is not in $\Gamma$. Let $C_i$ be the connected component of $G-Y$ that $c$ belongs to. Since $c\notin V(G')$, it must be the case that there is a constraint $c'\in V(G')$ such that $\sigma_i(c)=\sigma_i(c')$. However, this implies that the relation corresponding to $c'|_\tau$ is not in $\Gamma$, a contradiction. This completes the proof of the claim.
% Hence, we conclude that a set $S\subseteq \cV(G)\cap \cV(G')$ is a strong backdoor set of $G$ into $\CSP(\Gamma)$ if and only if it is also a strong backdoor set of $G'$ into $\CSP(\Gamma)$.
\end{proof}

%
%and that for every pair of vertices $u,v\in S^*_\omega \cup Z$, there is a $u$-$v$ path in $G$ with the internal vertices disjoint from $S^*_\omega$ if and only if there is such a path in $G'$.
%
%
%

\begin{claim}\label{clm:second} For any $S\subseteq Y$ and any $u,v\in S\cup Z$, there is a $u$-$v$ path in $G$ with the internal vertices disjoint from $S$ if and only if there is such a path in $G'$. 
	
\end{claim}

\begin{proof}
For the forward direction, let $P$ be a $u$-$v$ path in $G$ with the internal vertices disjoint from $S$. If $P$ is also present in $G'$ then we are done. Otherwise,
 consider a pair of vertices $x,y$ on $P$ such that $x,y\in Y$ and no vertex of the path $P$ \emph{between} $x$ and $y$ is present in $V(G')$. Then, it must be the case that $x,y\in N(C_i)$ for some $i\in [s]$. However, by the definition of $G'$, there is a redundant binary constraint whose scope is precisely the pair $x,y$. Hence, we can replace the segment of $P$ between $x$ and $y$ with the 2-length path through this constraint vertex. We repeat this for every such segment, concluding that there is a $u$-$v$ path in $G'$ which is internally disjoint from $S$. For the converse direction, we can use a similar argument where we replace segments of the path which use vertices of $V(G')\setminus V(G)$ with paths that pass through the connected components $C_1,\dots, C_s$ to get a $u$-$v$ \emph{walk} with the internal vertices disjoint from $S$ which also implies the presence of a path of the required kind. This completes the proof of the claim.
\end{proof}
We now proceed to complete the proof of the lemma using the above claims. 
Suppose that $(G\oplus G_1)$  has a strong backdoor $S$ of width at most $k$ into $\CSP(\Gamma)$. Due to Lemma \ref{lem:local_cut_paste} and the fact that $Y$ contains $S^*_\omega$ for every $\omega\in \mathfrak{S}((G,Z),h^*(k,|Z|),k)$, we may assume that $S\cap V(G)$, denoted $S^*$, is a subset of $Y$ and hence $S^*\subseteq \cV(G')$. This implies that $S\subseteq \cV(G'\oplus G_1)$. We now argue that $S$ is also a  a strong backdoor of $(G'\oplus G_1)$ of width at most $k$ into $\CSP(\Gamma)$.

Claim \ref{clm:first} implies that $S$ is indeed a strong backdoor of $(G'\oplus G_1)$. We now argue that $\torso_{G\oplus G_1}(S)$ is isomorphic to $\torso_{G'\oplus G_1}(S)$. We do this by showing that for every pair of vertices $u,v\in S$, $(u,v)\in E(\torso_{G\oplus G_1}(S))$ if and only if $(u,v)\in E(\torso_{G'\oplus G_1}(S))$. Equivalently, we argue that for every pair of vertices $u,v\in S$, there is a $u$-$v$ path in $G\oplus G_1$ disjoint from $S$ if and only if there is such a path in $G'\oplus G_1$. But this is a straightforward consequence of Claim \ref{clm:second}. This completes the argument in the forward direction. For the converse direction, suppose that $(G'\oplus G_1)$  has a strong backdoor $S'$ of width at most $k$ into $\CSP(\Gamma)$. It follows from the definition of $G'$ that $S\cap \cV(G')\subseteq Y$.

We now argue that $S$ is a a strong backdoor of $(G\oplus G_1)$ of width at most $k$ into $\CSP(\Gamma)$. Once again, Claim \ref{clm:first} implies that $S$ is indeed a strong backdoor and Claim \ref{clm:second} implies that for every pair of vertices $u,v\in S$, $(u,v)\in E(\torso_{G\oplus G_1}(S))$ if and only if $(u,v)\in E(\torso_{G'\oplus G_1}(S))$, implying that $\torso_{G\oplus G_1}(S)$ is isomorphic to $\torso_{G'\oplus G_1}(S)$. Clearly, choosing $\iota_{\sd}(k,t)$ to be the maximum of $\Upsilon(t,h^*(k,t),k)\cdot (k+1)+t$ and $max_{i\in [s]}\{|\CCC_i\cup \hat \CCC_i\cup  \CCC'_i\cup \hat \cV_i \cup N(C_i)|\}$ (which is easily seen to be bounded by a function of $k$ and $t$)  ensures that the treewidth of $G'$ is bounded by $\iota_{\sd}(k,t)$, concluding the proof of the lemma.
\end{proof}
}

\lv{Combining Lemma \ref{lem:reduction_strong} and Lemma \ref{lem:strong_rep_bound}, we get the desired lemma (stated below).

\begin{lemma}
	\label{lem:strong_rep_bound_main}
	There is a function $\alpha_{\sd}:\naturals^2\to \naturals$ such that for every $k,t\in {\mathbb N}$ and  $t$-boundaried incidence graph $(G,Z)$ there is a $t$-boundaried incidence graph $(G',Z')$ of size at most $ \alpha_{\sd}(k,t)$ such that $(G,Z)\sim_{t,k}(G',Z')$. 
\end{lemma}

As a consequence of Lemma~\ref{lem:strong_rep_bound_main}, we also get.
}
\lv{\begin{lemma}
\label{lem:compute reps}
	Let $k,t\in \mathbb{N}$. There exists a set $\mathfrak{F}_s(t,h^*(k,t),k)$ of at most $2^{{\alpha_{\sd}(k,t)\choose 2}} \cdot {\alpha_{\sd}(k,t)\choose t}\cdot t! \cdot \alpha_{\sd}(k,t) ^{2^{\cD^\rho(\Gamma,k)}} \cdot \alpha_{\sd}(k,t)^{\rho(\Gamma,k)!}$ $t$-boundaried CSP instances that contains a $t$-boundaried CSP instance from every equivalence class of $\sim_{t,k}$. Furthermore, given $k$ and $t$, the set $\mathfrak{F}_s(t,h^*(k,t),k)$ can be computed in time  $\bigoh(|\mathfrak{F}_s(t,h^*(k,t),k)|)$.
\end{lemma}
}

\sv{\begin{lemma}[$\star$]
\label{lem:compute reps}
	Let $k,t\in \mathbb{N}$. There exists a computable function $\mathfrak{y}(t,k)$ and a set $\mathfrak{F}_s(t,h^*(k,t),k)$ of at most $\mathfrak{y}(t,k)$
%	2^{{\alpha_{\sd}(k,t)\choose 2}} \cdot {\alpha_{\sd}(k,t)\choose t}\cdot t! \cdot \alpha_{\sd}(k,t) ^{2^{\cD^\rho(\Gamma,k)}} \cdot \alpha_{\sd}(k,t)^{\rho(\Gamma,k)!}$ 
$t$-boundaried CSP instances that contains a $t$-boundaried CSP instance from every equivalence class of $\sim_{t,k}$. Furthermore, given $k$ and $t$, the set $\mathfrak{F}_s(t,h^*(k,t),k)$ can be computed in time  $\bigoh(|\mathfrak{F}_s(t,h^*(k,t),k)|)$.
\end{lemma}
}

\lv{
\begin{proof}
	The first term is a bound on the  number of graphs on at most $\alpha_{\sd}(k,t)$ vertices, the second term bounds the number of possible choices of the boundary variables with the third term corresponding to the possible labelings of the $t$ boundary variables. The next term correspond to all ways of assigning relations of $\Gamma$ (which must have arity at most $\rho(\Gamma,k)$) to the constraint vertices and the final term corresponds to choosing the order of the variables in the scope of each constraint. It is straightforward to see that this captures all possible CSP instances of size at most $\alpha_{\sd}(k,t)$ that appear in any equivalence class of $\sim_{t,k}$ and since any equivalence class has an instance not exceeding this size bound, the lemma follows.	
\end{proof}
}

%\newpage
\section{The {\FPT} Algorithm for {\sc Width Strong-$\CSP(\Gamma)$ Backdoor Detection}}
\label{sec:algo}

An often-used approach in the design of {\FPT} algorithms for graph problems is that of finding a sufficiently small separator in the graph and then reducing one of the sides. In the technique of `recursive understanding' introduced by Grohe et al. \cite{GroheKMW11}, this is achieved by performing this step \emph{recursively} until we arrive at a separator where the side we want to reduce has certain connectivity-based structure using which we can find a way reduce it without recursing further. This approach has been combined with various problem specific reduction rules at the bottom to obtain parameterized algorithms for several well-studied problems. These include the $k$-{\sc way Cut} problem, solved by 
Kawarabayashi and Thorup \cite{KawarabayashiT11}, {\sc Steiner Cut} and {\sc Unique Label Cover} -- both solved by Chitnis et al. \cite{ChitnisCHPP12}.
In this section, we will employ this technique to design our algorithm for {\sc Width Strong-$\CSP(\Gamma)$ Backdoor Detection}. We begin by defining a notion of \emph{nice} instances which basically capture the kind of instances we will be dealing with at the bottom of our recursion.

\subsection{Solving Nice Instances}

\lv{\begin{lemma}}
\sv{\begin{lemma}[$\star$]}
\label{lem:small_backdoor_nice}
There is a function $\mathfrak{Z}:\naturals^2\to \naturals$ and an algorithm that, given a CSP instance $\III$ with incidence graph $G$ and positive integers $\beta,k\in \naturals$, runs in time $\bigoh(\mathfrak{Z}(\beta,k)|G|^2)$ and either computes a strong backdoor into $\CSP(\Gamma)$ of width at most $k$ or correctly concludes that $\III$ has no backdoor set $X$ of width at most $k$ that satisfies the following properties: 
	\begin{enumerate*}[label={\normalfont{\textbf{(\arabic*)}}}]
	\item $G-X$ has exactly one connected component $C$ of size at least $\beta+1$.
	\item $|V(G)\setminus N[C]|\leq \beta$
	\end{enumerate*}   
\end{lemma}

\lv{
\begin{proof}
The algorithm begins by enumerating all minimal strong backdoor sets into $\CSP(\Gamma)$ of size at most $\beta$. Since the arity of every constraint in $\III$ can be assumed to be bounded by $\rho=\rho(\Gamma,k)$, there exists a set $\cY=\{Y_1,\dots, Y_r\}$ containing all minimal strong backdoor sets $\CSP(\Gamma)$ of size at most $\beta$, where $r\leq \rho^\beta$, and furthermore $\cY$ can be computed in time at most $\bigoh(\rho^\beta\cdot |V(G)|^2)$. For each such backdoor set, we test whether it has width at most $k$ (by computing the torso and then computing its treewidth) and whether it creates a component $C$ satisfying the required properties; if this is the case for at least one of these strong backdoor sets, then we are done. So, suppose that this is not the case and furthermore suppose that $\III$ has a backdoor set $X$ of width at most $k$ that satisfies the stated properties. In particular, observe that since $|V(G)\setminus N[C]|\leq \beta$, such $X$ must have cardinality at most $\beta$, and hence there exists at least one element of $\cY$ which is a subset of $X$. Note that $X$ need not be a minimal strong backdoor; it could contain additional elements which separate the instance so as to ensure that $X$ has the required width. Let us branch over $\cY$, knowing that there exists some $Y_i$ such that $Y_i\subseteq X$; assume without loss of generality that $Y_i=Y_1$. Since $C\cap X=\emptyset$, it follows that $Y_1\cap C=\emptyset$.

If $Y_1\supseteq N(C)$, then we can detect $X$ as follows. We test whether there exists a single component $C$ in $G-Y_1$ if size at least $\beta+1$ and whether $|V(G)\setminus N[C]|\leq \beta$. We then branch over the at most $2^\beta$ subsets $Q$ of $V(G)\setminus N[C]$ and test whether $X=Q\cup Y_1$ satisfies the lemma. 

%by just iterating over every component of $G-Y_1$, and checking whether taking all variables outside this component gives us a strong backdoor set of width at most $k$. 
Hence, we may assume that there is a vertex of $N(C)$ which is not already in $Y_1$. We now argue that if a variable in $N(C)$ is not in the scope of a constraint outside $C$, then we may assume that this variable is already in $Y_1$. This is because removing this variable from the strong backdoor set $X$ will not increase the width of $X$. Hence, its sole purpose in $X$ is to reduce a constraint in $C$ to  $\CSP(\Gamma)$, which allows us to assume that it is in $Y_1$ since $Y_1$ is a minimal strong backdoor set into $\CSP(\Gamma)$. 
At this point, we will design a branching algorithm that attempts to find variables in $N(C)\setminus Y_1$. Note that, as argued above, we may assume that every variable in $N(C)\setminus Y_1$ is in the scope of a constraint outside $C$.

Now, observe that since $Y_1$ does not have width at most $k$ it must be the case that $|Y_1|>k+1$. Also observe that since $X$ has width at most $k$, it must hold that $|N(C)|\leq k+1$ (otherwise the torso would contain a clique of size $k+2$). That is, there is a variable $y\in Y_1$ such that $y\notin N(C)$. However, this implies that there is a connected subgraph in $G$ that contains $y$, has size at most $\beta$, and has a neighborhood at size most $k+1$ such that the neighborhood intersects $N(C)\setminus Y_1$. Hence, we simply branch over the neighborhoods of all such connected subgraphs containing a vertex of $Y$ in order to locate the vertices in $N(C)\setminus Y_1$. By \cite{FominV12} there exist at most $|Y_1|\cdot {{\beta +k+1}\choose \beta}$ such subgraphs and each of them has a neighborhood of size at most $k+1$. Hence, we only need to branch over a set of at most $|Y_1|\cdot {{\beta +k+1}\choose \beta} \cdot (k+1)$ variables. Finally, once we have a subset of $X$ that contains $N(C)$, the argument for the case when $N(C)\subseteq Y_1$ can be applied to verify the existence of $X$.	If every branch of this algorithm concludes that there is no strong backdoor of  width at most $k$ then we can correctly conclude that there is no width-$k$ backdoor that satisfies the specified properties. Setting $\mathfrak{Z}(\beta,k)=(\rho^\beta(\beta+k+1)^\beta (k+1))^{k+1}$ completes the proof of the lemma.
\end{proof}
}

We now give the definition of `nice' instances. Generally speaking, these are instances which fall into either the bounded `classical' treewidth case or bounded backdoor size case. 
\lv{The formal definition is provided below.}

%\lv{
\begin{definition}
	Let $k,\beta\in \mathbb{N}$ and $\III$ be a CSP instance with incidence graph $G$. We say that $\III$ is $(\beta,k)$-\emph{nice} if  
	\lv{
	\vspace{-0.2cm}
	\begin{itemize}
		\item $tw(G)\leq \beta+k$ and/or
		\item if $\III$ has some strong backdoor set of width at most $k$, then it also has a strong backdoor set $X$ of width at most $k$ which satisfies the following properties:
%			\sv{\vspace{-0.2cm}}
		\begin{itemize}
%		\item $X$ has width at most $k$,
	\item $G-X$ has exactly one connected component $C$ of size at least $\beta+1$, and
	\item $|V(G)\setminus N[C]|\leq \beta$.
	\end{itemize} 
		\sv{\vspace{-0.2cm}}
			\end{itemize}
			}
			\sv{ $tw(G)\leq \beta+k$ or
	if $\III$ has some strong backdoor set of width at most $k$, then it also has a strong backdoor set $X$ of width at most $k$ such that 			$G-X$ has exactly one connected component $C$ of size at least $\beta+1$, and $|V(G)\setminus N[C]|\leq \beta$.
			}
\end{definition}
%}
%\sv{
%	Let $k,\beta\in \mathbb{N}$. Let $\III$ be a CSP instance with incidence graph $G$. We say that $\III$ is $(\beta,k)$-\emph{nice} if one of the following properties hold: 
%	\begin{enumerate*}[label={\textbf{(\arabic*)}}]
%		\item $tw(G)\leq \beta+k$, or
%		\item if $\III$ has a strong backdoor set of width at most $k$ then it also has a strong backdoor set $X$ which satisfies the following properties:
%		\begin{enumerate*}[label={{(\alph*)}}]
%		\item $X$ has width at most $k$,
%	\item $G-X$ has exactly one connected component $C$ of size at least $\beta+1$, and
%	\item $|V(G)\setminus N[C]|\leq \beta$.
%	\end{enumerate*} 
%			\end{enumerate*}
%}
We now formally show that given a $(\beta,k)$-nice incidence graph, one can detect strong backdoor sets of small width in {\FPT} time parameterized by $\beta+k$. This will later be used to compute small representatives of large boundaried CSP instances (specifically, in Lemma~\ref{lem:stucture_minimal_rep}).

\begin{lemma}\label{lem:nice algorithm}
 There is a function $\hat{\mathfrak{X}}:\naturals \to \naturals$ and an algorithm that, given $\beta,k\in \mathbb{N}$, a $(\beta,k)$-nice CSP instance $\III$ with the incidence graph $G$, runs in time $\bigoh(\hat{\mathfrak{X}}(\beta+k)\vert G\vert^2)$ and either computes a strong backdoor set into $\CSP(\Gamma)$ of width at most $k$ or correctly detects that such a set does not exist.
\end{lemma}

\begin{proof}
If $tw(G)\leq \beta+k$, then we can solve the problem directly by applying Courcelle's theorem \cite{Courcelle90b}, as follows. First, recall that the arity of any constraint which appears in the CSP instance $\psi(G)$ is upper-bounded by $k$ plus the maximum arity of relations in $\Gamma$. Hence we can assume that the number of relations which appear in the constraints of $\psi(G)$ is bounded by a function of $k$, and we can think of $G$ as having vertex labels which specify which relation is used in each constraint vertex and edge labels which specify the order in which variables appear in the incident constraint. Second, for a $j$-ary relation $R$ which appears in a constraint $C$ in $\psi(G)$, we say that a subset $\alpha$ of $\{1,\dots,j\}$ is a \emph{valid choice} if the variables occurring in positions $\alpha$ in $C$ form a strong backdoor for $\{C\}$ into $\CSP(\Gamma)$. Note that the set of valid choices for all of the relations which occur in a constraint in $\psi(G)$ can be precomputed in advance. Then the problem can be formulated in Monadic Second Order logic with a sentence stating the following: there exists a set $T$ of variables such that $(1)$ for each constraint $C$ with label $R$ it holds that the edges between $T$ and $C$ correspond to a valid choice for $R$, and $(2)$ the torso of $T$ does not contain any of the forbidden minors for treewidth at most $k$. Indeed, condition $(1)$ ensures that $T$ forms a backdoor to $\CSP(\Gamma)$ and condition $(2)$ ensures that $T$ has width at most $k$.

Otherwise, we execute the algorithm of Lemma \ref{lem:small_backdoor_nice} that runs in time $\bigoh(\hat{\mathfrak{Z}}(\alpha)\vert G\vert^2)$. 
%It is straightforward to verify that the same algorithm with a slight modification can be used to check for the presence of a strong backdoor set of the required kind which also happens to contain the set $Y$. 
 The function $\mathfrak{X}$ is obtained from the function $\mathfrak{Z}$ and the dependence of the algorithm on $\beta+k$ in the case of bounded treewidth.
%
% Setting $\hat{\mathfrak{X}}(\alpha+k)=max\{\hat{\mathfrak{Z}}(\alpha+k)$, function of $\alpha+k$ we get from Courcelle's theorem$\}$ completes the proof of the lemma.
	\end{proof}

\subsection{Computing a Minimal Representative}
In this subsection, we show that if  a $t$-boundaried instance has a certain guarantee on the \mbox{(non-)}existence of a small separator separating two large parts of the instance from each other, then we can compute a $t$-boundaried instance of bounded size which is equivalent to it under the relation $\sim_{t,k}$.

\begin{definition}
Let $G$ be an incidence graph and $(A,S,B)$ be a partition of $V(G)$ where $S\subseteq \cV(G)$ and $N(A),N(B)\subseteq S$. We call $(A,S,B)$ a $(q,k)$-\emph{separation} if $S$ has size at most $k$, and $A$ and $B$ have size at least $q$.
\end{definition}

\lv{\begin{lemma}}
\sv{\begin{lemma}[$\star$]}
\label{lem:unbreakable_implies_nice}
Let $G$ be the incidence graph of a CSP instance $\III$.
	If $G$ has no $(q,k+1)$-separation then $\III$ is $(q,k)$-nice.
\end{lemma}

\lv{
\begin{proof}
We will prove the lemma by showing that either $\III$ is $(q,k)$-nice or $G$ contains a $(q,k+1)$-separation.
Let $X$ be a hypothetical strong backdoor set of width at most $k$ for the CSP instance $\III$.
% where $G'=G\oplus H$ is the incidence graph of $\III_1\oplus \III_2$. 
 Let $C$ be the largest component of $G-X$, let $S=N(C)$ and $D=V(G)\setminus (C\cup S)$. Since $X$ has width at most $k$, it follows that every connected component of $G-X$ has at most $k+1$ neighbors in $X$ and in particular, $|S|\leq k+1$. 
		
		Observe that if $C$ has size at most $q$, then we are done since $G$ has treewidth at most $q+k$. Indeed, we can obtain a tree-decomposition for $G$ of this width by starting with a tree-decomposition of width at most $k$ for $\torso_{G}(X)$ and then creating a new bag for each connected component of $G-X$ which contains the vertices of this component along with its neighborhood in $X$. 
		
Otherwise, observe that if $D$ has size at most $q$, then $\III$ is $(q,k)$-nice and the lemma holds. So, suppose for a contradiction that $D$ has size greater than $q$. This implies that $(C,N(C),D)$ is a $(q,k+1)$-separation in $G$, and hence the lemma holds.
%Since $|V(H)|\leq \alpha_{\sd}(k,2(k+1))$,  Lemma \ref{lem:unbreakable_implication} implies that $G$ contains a $(8^q,k+1)$-separation which is a contradiction to the premise of the lemma. This completes the argument for the first statement.
	\end{proof}
	}
	
\lv{\begin{lemma}}
\sv{\begin{lemma}[$\star$]}
\label{lem:local_unbreakable_implies_global_unbreakable}
	Let $t\in \naturals$ and $\III_1$ be a $t$-boundaried CSP instance with $t$-boundaried incidence graph $(G,Z)$. Let $k,q\in \naturals$ be such that $G$ does not admit a $(8^q,k+1)$-separation, and let $(H,J)$ be the $t$-boundaried incidence graph of a $t$-boundaried CSP instance $\III_2$ such that the size of $V(H)$ is at most some $r\in \naturals$. Then the incidence graph $G\oplus H$ corresponding to the instance $\III_1\oplus \III_2$ has no $(8^q+r,k+1)$-separation.
\end{lemma}

\lv{
\begin{proof}
	Suppose to the contrary that the graph $G'=G\oplus H$ admits a $(8^q+r,k+1)$-separation $(A,S,B)$. Let $A'=A\setminus (V(H)\setminus J)$, $B'=B\setminus (V(H)\setminus J)$ and $S'=S\setminus (V(H)\setminus J)$. We claim that $(A',S',B')$ is a $(8^q,k+1)$-separation in $G$.

	Observe that it follows from the definition of $A',B', S'$ that $V(G)=A' \uplus S' \uplus B'$. Furthermore, since $V(H)\leq r$, it follows that $|A'|,|B'|\geq 8^q$ and since $S'\subseteq S$, it follows that $|S'|\leq k+1$. Since $(A,S,B)$ is a separation in $G'$, it follows that there is no edge in $G'$ with one endpoint in $A$ and one in $B$. Since $A'\subseteq A$ and $B'\subseteq B$, we conclude that there is no edge in $G$ with one endpoint in $A'$ and one in $B'$. This implies that $(A',S',B')$ is indeed a $(8^q,k+1)$-separation in $G$, a contradiction to the premise of the lemma.
\end{proof}
}

For the following lemma, recall the definition of the set $\mathfrak{F}_s(t,h^*(k,t),k)$ (Lemma~\ref{lem:compute reps}). 
\sv{The proof relies on Lemmas~\ref{lem:local_unbreakable_implies_global_unbreakable}, \ref{lem:unbreakable_implies_nice}, and~\ref{lem:nice algorithm}.}

\lv{\begin{lemma}}
\sv{\begin{lemma}[$\star$]}
\label{lem:stucture_minimal_rep}
Let $t\in \naturals$ and $\III_1$ be a $t$-boundaried CSP instance with incidence graph $G$ and boundary $Z$. Further, let $k,q\in \naturals$ be such that $t\leq 2(k+1)$, $|V(G)|>q$, and $G$ does not admit a $(8^q,k+1)$-separation. Let $(H,J)$ be the $t$-boundaried incidence graph of a $t$-boundaried CSP instance $\III_2$ in $\mathfrak{F}_s(t,h^*(k,t),k)$. Then the instance $\III_1\oplus \III_2$ is $(8^q+\alpha_{\sd}(k,2(k+1)),k)$-nice. Furthermore,
if $q=\alpha_{\sd}(k,2(k+1))$ then one can compute in time $\bigoh(\mathfrak{M}(k)|G|^2)$ a $t$-boundaried CSP instance $\III_1^*$ of size at most $q$ such that 	$\III_1\sim_{t,k}\III_1^*$, for some function $\mathfrak{M}$.
	\end{lemma}
	
\lv{
	\begin{proof}
		
		It follows from Lemma \ref{lem:local_unbreakable_implies_global_unbreakable} in conjunction with Lemma~\ref{lem:strong_rep_bound_main} that the graph $G\oplus H$ corresponding to the instance $\III_1\oplus \III_2$ has no $(8^q+\alpha_{\sd}(k,2(k+1)),k+1)$-separation. By Lemma \ref{lem:unbreakable_implies_nice} this implies that $\III_1\oplus \III_2$ is $(8^q+\alpha_{\sd}(k,2(k+1)),k)$-nice. This completes the argument for the first statement.

We now argue the second statement. For each instance $\III\in \mathfrak{F}(t,h^*(k,t),k)$, we construct the instance $\III_1\oplus \III$ and 
%for every subset $Y$ of the boundary variables of $\III_1\oplus \III$,  
execute the algorithm of Lemma \ref{lem:nice algorithm} to check for the existence of a strong backdoor set of width at most $k$. 
%that contains $Y$  on this instance. 
Following this, for each pair of instances $\III_a,\III_b\in \mathfrak{F}(t,h^*(k,t),k)$ we do the same on the instance $\III_a\oplus \III_b$. We define $\III^*$ to be that instance in $\mathfrak{F}(t,h^*(k,t),k)$ with the property that $\III_1\oplus \III$ has a strong backdoor into $\CSP(\Gamma)$ of width at most $k$ 
%containing the set $Y$  
if and only if $\III^*\oplus \III$ has a strong backdoor into $\CSP(\Gamma)$ of width at most $k$ 
%and containing $Y$, 
for every $\III\in \mathfrak{F}(t,h^*(k,t),k)$. We use the bounds stated in Lemma \ref{lem:nice algorithm} and the fact that each $\III\in \mathfrak{F}(t,h^*(k,t),k)$ has size bounded by $\alpha_{\sd}(k,2(k+1)$ to appropriately define the function $\mathfrak{M}$.  This completes the proof of the lemma.
\end{proof}
}

%	\begin{lemma}\label{lem:compute_min_rep}
%%		There is a function $...$ and an algorithm that, given $k\in \naturals$ and a $t$-boundaried incidence graph $(G,Z)$ corresponding to a $t$-boundaried CSP instance $\III_1$ where $t\leq 2(k+1)$ with the additional property that  $|V(G)|>q$ and $G$ does not admit a $(8^q,k+1)$-separation where $q=\alpha_{\sd}(k,2(k+1))$,
%%		runs in time $...$ and returns a $t$-boundaried CSP instance $\III_2$ of size at most $q$ such that $\III_1\sim_{t,k} \III_2$.
%
%Let $t\in \naturals$ and $\III_1$ be a $t$-boundaried CSP instance with incidence graph $G$ and boundary $Z$. Further, let $k,q\in \naturals$ be such that, $t\leq 2(k+1)$ and $G$ does not admit a $(8^q,k+1)$-separation. 
%		
%		
%		
%	\end{lemma}
%	
%	
%	\begin{proof}
%		
%	\end{proof}

%\newpage
\subsection{Solving the Problem via Recursive Understanding}
In this subsection, we complete our algorithm for {\sc Width Strong-$\CSP(\Gamma)$ Backdoor Detection} by describing the recursive phase of our algorithm and the way we utilize the subroutines described earlier to solve the problem. We note that variants of Lemma \ref{lem:fast compute good seps}, Lemma \ref{lem:find_good_sep} and Lemma \ref{lem:compute_extremal_cut} are well-known in literature (see for example \cite{ChitnisCHPP12}). However the parameters involved in these lemmas are specific to the application. Furthermore, our proofs are simpler and avoid the color coding technique employed in \cite{ChitnisCHPP12}.

\lv{\begin{lemma}}
\sv{\begin{lemma}[$\star$]}
\label{lem:fast compute good seps}
	There is an algorithm that, given an incidence graph $G$ and $q,k\in \mathbb{N}$, runs in time $\bigoh((2q)^k\cdot \vert G \vert^2)$ and either computes a $(q,k)$-separation or concludes correctly that there is no  $(q,k)$-separation $(A,S,B)$ where $A$ and $B$ are connected.
\end{lemma}

\lv{
\begin{proof}
We describe a branching algorithm which we analyze using $k$ as the measure. 
	We go over all pairs of constraint vertices $u,v$ in $G$ and for each pair we test if there is a $(q,k)$-separation $(A,S,B)$ where $u\in A$ and $v\in B$. This is done as follows.
	
	We pick an arbitrary connected $q$-vertex subgraph $H_u$ of $G$ containing $u$ and an arbitrary connected $q$-vertex subgraph $H_v$ of $G$ containing $v$. We now perform a $2q$-way branching where in each branch, we pick a unique variable vertex $x$ in $H_u\cup H_v$ and recursively try to compute a $(q,k-1)$-separation $(A,S',B)$ in $G-x$ where $u\in A$ and $v\in B$. 
	 If we do not find the required $(q,k-1)$-separation in any of these branches, then it must be the case that either there is no $(q,k)$-separation of the kind we are looking for at all or $S$ is disjoint from $H_u\cup H_v$. Observe that in the latter case it follows that $H_u\cap H_v=\emptyset$ and furthermore that $V(H_u)\subseteq A$ and $V(H_v)\subseteq B$. In this case, we simply need to test whether there is  \emph{any} set of size at most $k$ which is disjoint from $H_u\cup H_v$ and separates $V(H_u)$ and $V(H_v)$. This can be achieved by a simple max-flow computation, thus completing the proof of the lemma.
	 	\end{proof}
}

\sv{Next, we use Lemma~\ref{lem:fast compute good seps} to obtain the final ingredient for our algorithm.}

\sv{\begin{lemma}[$\star$]}
\lv{\begin{lemma}}
\label{lem:find_good_sep}
 There is an algorithm that, given an incidence graph $G$ and $q,k\in \mathbb{N}$, runs in time $\bigoh((q+k)^k\vert G\vert^2)$ and either computes a $(q,k)$-separation in $G$ or correctly concludes that $G$ has no $(8^q,k)$-separation.
 
\end{lemma}

\lv{
\begin{proof} The algorithm has 2 phases. 
%The first phase is executed under the assumption that $G$ has a $(q,k)$-separation $(A,S,B)$ where $A$ and $B$ are connected. For this we 
The first phase executes the algorithm of Lemma \ref{lem:fast compute good seps}.
If the algorithm does not succeed in computing a $(q,k)$-separation in the first phase, the second phase is executed as follows. For every vertex $v\in \cV(G)$, we compute the family $\cH_v$ which denotes the set of connected subgraphs of $G$ of size at most $q$ which contain $v$ and have a neighborhood of size at most $k$. The size of each such family is at most ${q+k\choose k}$ and can be computed in time $ \bigoh({q+k\choose k}\vert G\vert)$ \cite{FominV12}. We define ${\cal H}=\bigcup_{v\in V(G)}{\cal H}_v$.

Having computed the families, we group the connected subgraphs by their neighborhood. That is, $H_1$ and $H_2$ in $\cal H$ are in the same class iff $N(H_1)=N(H_2)$. We then check if there is a variable set $X$ such that the set $W(X)=\SB v ~|~ v\in V(H), H\in {\cal H}: N(H)=X \SE$ has size at least $3q$. Suppose such a set $X$ exists. Since every $H\in {\cal H}$ has size at most $q$, we can construct a $(q,k)$ separation from these components. If such a set $X$ does not exist, then we return that $G$ does not have a $(8^q,k)$-separation. We now prove that this algorithm is indeed correct. For this, we prove that if $G$ has a $(8^q,k)$-separation then there is a variable set $X$ such that the set $W(X)=\SB v~|~v\in V(H), H\in {\cal H}: N(H)=X \SE$ has size at least $3q$.

Suppose that $(A,S,B)$ is a $(8^q,k)$-separation in $G$. We now group the connected components of $G[A]$ according to their neighborhood in $S$.
Since $\vert A\vert > 8^q$, there is a set $S'\subseteq S$ such that the connected components of $G[A]$ with $S'$ as the neighborhood contain at least $\log (8^q)=3q$ vertices of $A$. 
%Futhermore, since every connected subgraph of $G-S'$ has size atmost $q$, we can construct a $(2q,k)$-separation using these connected subgraphs and $S'$. 
This completes the proof of the lemma.
	\end{proof}
	}

\begin{observation}\label{obs:either_or}
	Let $(G,Z)$ be a $t$-boundaried graph with $|V(G)|>q$ and $t\leq 2(k+1)$ and let $(A,S,B)$ be a $(q,k+1)$-separation in $G$. Then, one of the pairs $(G[A\cup S],S\cup (Z\cap A))$ or  $(G[B\cup S],S\cup (Z\cap B))$ is a $t'$-boundaried graph with $t'\leq 2(k+1)$.
	\end{observation}

	\begin{lemma}\label{lem:compute_extremal_cut}
		There is an algorithm that, given a $t$-boundaried graph $(G,Z)$ with  $|V(G)|>q$ and $t\leq 2(k+1)$, in time $\bigoh((q+k)^k\vert G\vert^3)$ returns a $t'$-boundaried graph $(G',Z')$ where $G'$ is a subgraph of $G$, (a) $|V(G')|>q$, 
(b) $t'\leq 2k+1$, and
(c) $G'$ has no $(8^q,k+1)$-separation.
	%\sv{
%\begin{enumerate*}[label={\normalfont{\textbf{(\arabic*)}}}]
%	\item $|V(G')|>q$, 
%	\item $t'\leq 2k+1$, and
%	\item $G'$ has no $(8^q,k+1)$-separation.
%	\end{enumerate*}
%}
	\end{lemma}

	\begin{proof} We begin by executing the algorithm of Lemma \ref{lem:find_good_sep}. If this algorithm returns that $G$ has no $(8^q,k+1)$-separation then we terminate the algorithm and return the graph $(G,Z)$ itself. Otherwise, let $(X,S,Y)$ be the $(q,k+1)$-separation returned by this algorithm.  By Observation \ref{obs:either_or}, we may assume w.l.o.g.
	%without loss of generality 
	that $(G[X\cup S],S\cup (Z\cap X))$ is a $t''$-boundaried graph where $t''\leq 2(k+1)$. We now set $G:=G[X\cup S]$, $Z:=S\cup (Z\cap X)$ and recurse. The depth of recursion is clearly bounded by the size of the input graph. Since each step takes time $\bigoh((q+k)^k\vert G\vert^2)$, the lemma follows.		
	\end{proof}

\noindent
\textbf{Algorithm for the Decision version of Theorem \ref{thm:find_backdoor}.} Let $\III$ be the given input CSP instance and let $G$ be its incidence graph. We begin by setting $q=\alpha_{\sd}(k,2(k+1))$, choosing the boundary $Z$ to be the empty set and then executing the algorithm of Lemma \ref{lem:compute_extremal_cut} to compute a $t$-boundaried graph $(G',Z')$ where $G'$ is a subgraph of $G$, $|V(G')|>q$ and $t\leq 2(k+1)$ such that $G'$ has no $(8^q,k+1)$-separation. 

Next, we invoke Lemma \ref{lem:stucture_minimal_rep} on the corresponding CSP instance, say  $\III'$, to compute in time $\bigoh(\mathfrak{M}(k)|G|^2)$ a $t$-boundaried CSP instance $\III''$ such that $\III'\sim_{t,k}\III''$. We then set $\III=\III''\oplus (\psi(G-(V(G')\setminus Z')))$ and recursively check for the presence of a strong backdoor set of width at most $k$ for this instance. Since we strictly reduce the size of the instance in each step, the depth of the recursion is bounded linearly in the size of the initial input, implying {\FPT} running time.\medskip
% compute a strong backdoor set of width at most $k$ for this instance. The last statement of Lemma \ref{lem:stucture_minimal_rep} ensures that we can recover a strong backdoor set of width at most $k$ for the input instance when given one for the reduced instance. 
% This completes the proof of the theorem.

\lv{
\noindent
\textbf{Computing a strong backdoor set of width at most $k$ given the decision algorithm.} Recall that the algorithm of Lemma \ref{lem:usingwidth} requires as input a strong backdoor set of width at most $k$. We use the self-reducibility of this problem in order to compute a strong backdoor set using the decision algorithm as a sub-routine.
Let $\III$ be the given $\CSP$ instance and $k$ be the given budget. We first show that for any $Y\subseteq \cV$, we can check in {\FPT} time whether $\III$ contains a strong backdoor set $X$ of width at most $k$ into $\CSP(\Gamma)$ such that $X\cap Y=\emptyset$. 

  Let $\III'$ be the instance defined from $\III$ as follows. For every variable $v\in Y$, we make $k+2$-copies of $v$, $\{v_1,\dots, v_{k+2}\}$ and for every constraint $c$ whose scope includes $v$, we make $k+2$ copies of this constraint $c_1,\dots, c_{k+2}$ with $v_i$ belonging to the scope of $c_i$.
  Finally, we add ${k+2} \choose 2$ tautological binary constraints, one on each pair of these copies of $v$. It is straightforward to verify that $\III$ has a strong backdoor of width at most $k$ disjoint from $Y$ if and only if $\III'$ has a strong backdoor of width at most $k$.   Given this subroutine, one can incrementally and in {\FPT} time construct a set $Y^*$ of variables such that $\cV(\III)\setminus Y^*$ is in fact a strong backdoor set for $\III$ of width at most $k$. This concludes the proof of Theorem \ref{thm:find_backdoor}.\qed
%  description of the algorithm  to compute a strong backdoor set of width at most $k$ if one exists.

\lv{\medskip}

%\noindent
%\textbf{Proof of Theorem \ref{thm:main theorem} and Theorem \ref{thm:main-counting}.} Both theorems follow from Theorem \ref{thm:find_backdoor} since one can compute a strong backdoor set of width at most $k$ (if it it exists) and then execute the algorithm of Lemma \ref{lem:usingwidth} to solve $\CSP$ and $\sharpCSP$.
\noindent
\textbf{Proof of Theorem \ref{thm:main main theorem} and Corollary~\ref{cor:counting}.} The proof follows from Theorem \ref{thm:find_backdoor} since one can compute a strong backdoor set of width at most $k$ (if it it exists) and then execute the algorithm of Lemma \ref{lem:usingwidth} to solve $\CSP$ and $\sharpCSP$.\qed
}

\sv{
\noindent
\textbf{Proof of Theorem \ref{thm:main main theorem} and Corollary~\ref{cor:counting}.} Using the self-reducibility of the problem and the algorithm for the decision variant of Theorem \ref{thm:find_backdoor},  one can compute a strong backdoor set of width at most $k$ (if it it exists). Following this, one can execute the algorithm of Lemma \ref{lem:usingwidth} to solve $\CSP$ and $\sharpCSP$.\qed
}

\section{Concluding Remarks}
We have introduced the notion of backdoor treewidth for \CSP\ and \sharpCSP\ by combining the two classical approaches of placing structural restrictions and language restrictions, respectively, on the input. Thus the presented results represent a new ``hybrid’’ approach for solving \CSP and \sharpCSP. Our main result, Theorem~\ref{thm:main main theorem}, is quite broad as it covers all tractable finite constraint languages combined with the graph invariant treewidth. This can be seen as the base case of a general framework which combines a specific graph invariant of the torso graph with a specific class of constraint languages. Therefore, we hope it will stimulate further research in this direction.

\enlargethispage*{10mm}
\bibliographystyle{abbrv}
\bibliography{literature}

\begin{thebibliography}{10}

\bibitem{ArnborgCPS93}
S.~Arnborg, B.~Courcelle, A.~Proskurowski, and D.~Seese.
\newblock An algebraic theory of graph reduction.
\newblock {\em J. ACM}, 40(5):1134--1164, 1993.

\bibitem{BessiereCarbonnelHebrardKatsirelosWalsh13}
C.~Bessiere, C.~Carbonnel, E.~Hebrard, G.~Katsirelos, and T.~Walsh.
\newblock Detecting and exploiting subproblem tractability.
\newblock In F.~Rossi, editor, {\em IJCAI 2013, Proceedings of the 23rd
  International Joint Conference on Artificial Intelligence, Beijing, China,
  August 3-9, 2013}. IJCAI/AAAI, 2013.

\bibitem{Bodlaender96}
H.~L. Bodlaender.
\newblock A linear-time algorithm for finding tree-decompositions of small
  treewidth.
\newblock {\em SIAM J. Comput.}, 25(6):1305--1317, 1996.

\bibitem{BodlaenderF96a}
H.~L. Bodlaender and B.~de~Fluiter.
\newblock Reduction algorithms for constructing solutions in graphs with small
  treewidth.
\newblock pages 199--208, 1996.

\bibitem{BodlaenderH98}
H.~L. Bodlaender and T.~Hagerup.
\newblock Parallel algorithms with optimal speedup for bounded treewidth.
\newblock {\em SIAM J. Comput.}, 27:1725--1746, 1998.

\bibitem{BodlaendervA01a}
H.~L. Bodlaender and B.~van Antwerpen-de Fluiter.
\newblock Reduction algorithms for graphs of small treewidth.
\newblock 167:86--119, 2001.

\bibitem{Bulatov06}
A.~A. Bulatov.
\newblock A dichotomy theorem for constraint satisfaction problems on a
  3-element set.
\newblock {\em J. of the ACM}, 53(1):66--120, 2006.

\bibitem{Bulatov11}
A.~A. Bulatov.
\newblock Complexity of conservative constraint satisfaction problems.
\newblock {\em ACM Trans. Comput. Log.}, 12(4):Art. 24, 66, 2011.

\bibitem{Bulatov13}
A.~A. Bulatov.
\newblock The complexity of the counting constraint satisfaction problem.
\newblock {\em J. of the ACM}, 60(5):Art 34, 41, 2013.

\bibitem{BulatovKrokhinJeavons01}
A.~A. Bulatov, A.~A. Krokhin, and P.~Jeavons.
\newblock The complexity of maximal constraint languages.
\newblock In J.~S. Vitter, P.~G. Spirakis, and M.~Yannakakis, editors, {\em
  Proceedings on 33rd Annual {ACM} Symposium on Theory of Computing, July 6-8,
  2001, Heraklion, Crete, Greece}, pages 667--674. {ACM}, 2001.

\bibitem{CarbonnelCooper16}
C.~Carbonnel and M.~C. Cooper.
\newblock Tractability in constraint satisfaction problems: a survey.
\newblock {\em Constraints}, 21(2):115--144, 2016.

\bibitem{CarbonnelCooperHebrard14}
C.~Carbonnel, M.~C. Cooper, and E.~Hebrard.
\newblock On backdoors to tractable constraint languages.
\newblock In {\em Principles and Practice of Constraint Programming - 20th
  International Conference, {CP} 2014, Lyon, France, September 8-12, 2014.
  Proceedings}, volume 8656 of {\em Lecture Notes in Computer Science}, pages
  224--239. Springer Verlag, 2014.

\bibitem{ChitnisCHPP12}
R.~H. Chitnis, M.~Cygan, M.~Hajiaghayi, M.~Pilipczuk, and M.~Pilipczuk.
\newblock Designing {FPT} algorithms for cut problems using randomized
  contractions.
\newblock In {\em 53rd Annual {IEEE} Symposium on Foundations of Computer
  Science, {FOCS} 2012, New Brunswick, NJ, USA, October 20-23, 2012}, pages
  460--469, 2012.

\bibitem{CohenJeavonsGyssens08}
D.~Cohen, P.~Jeavons, and M.~Gyssens.
\newblock A unified theory of structural tractability for constraint
  satisfaction problems.
\newblock {\em J. of Computer and System Sciences}, 74(5):721--743, 2008.

\bibitem{CooperCohenJeavons94}
M.~C. Cooper, D.~A. Cohen, and P.~G. Jeavons.
\newblock Characterising tractable constraints.
\newblock {\em Artificial Intelligence}, 65(2):347--361, 1994.

\bibitem{Courcelle90b}
B.~Courcelle.
\newblock The monadic second-order logic of graphs. {I}. recognizable sets of
  finite graphs.
\newblock {\em Inf. Comput.}, 85(1):12--75, 1990.

\bibitem{CramaEkinHammer97}
Y.~Crama, O.~Ekin, and P.~L. Hammer.
\newblock Variable and term removal from {Boolean} formulae.
\newblock {\em Discr. Appl. Math.}, 75(3):217--230, 1997.

\bibitem{Dalmau00}
V.~Dalmau.
\newblock A new tractable class of constraint satisfaction problems.
\newblock In {\em {AMAI}, 6th International Symposium on Artificial
  Intelligence and Mathematics, Fort Lauderdale, Florida, USA, January 5-7,
  2000}, 2000.

\bibitem{DalmauKolaitisVardi02}
V.~Dalmau, P.~G. Kolaitis, and M.~Y. Vardi.
\newblock Constraint satisfaction, bounded treewidth, and finite-variable
  logics.
\newblock In P.~V. Hentenryck, editor, {\em Principles and Practice of
  Constraint Programming - {CP} 2002, 8th International Conference, {CP} 2002,
  Ithaca, NY, USA, September 9-13, 2002, Proceedings}, volume 2470 of {\em
  Lecture Notes in Computer Science}, pages 310--326. Springer Verlag, 2002.

\bibitem{Fluiter97}
B.~de~Fluiter.
\newblock {\em Algorithms for Graphs of Small Treewidth}.
\newblock PhD thesis, Utrecht University, 1997.

\bibitem{Diestel12}
R.~Diestel.
\newblock {\em Graph Theory, 4th Edition}, volume 173 of {\em Graduate texts in
  mathematics}.
\newblock Springer, 2012.

\bibitem{Farnqvist12}
T.~F{\"{a}}rnqvist.
\newblock Constraint optimization problems and bounded tree-width revisited.
\newblock In N.~Beldiceanu, N.~Jussien, and E.~Pinson, editors, {\em
  Integration of {AI} and {OR} Techniques in Contraint Programming for
  Combinatorial Optimzation Problems - 9th International Conference, {CPAIOR}
  2012, Nantes, France, May 28 - June1, 2012. Proceedings}, volume 7298 of {\em
  Lecture Notes in Computer Science}, pages 163--179. Springer Verlag, 2012.

\bibitem{FederVardi98}
T.~Feder and M.~Y. Vardi.
\newblock The computational structure of monotone monadic {SNP} and constraint
  satisfaction: A study through datalog and group theory.
\newblock {\em SIAM J. Comput.}, 28(1):57--104, 1998.

\bibitem{FellowsL89}
M.~R. Fellows and M.~A. Langston.
\newblock An analogue of the {M}yhill-{N}erode theorem and its use in computing
  finite-basis characterizations (extended abstract).
\newblock In {\em FOCS}, pages 520--525, 1989.

\bibitem{RamanujanLokshtanovFominSaurabhMisra15}
F.~V. Fomin, D.~Lokshtanov, N.~Misra, M.~S. Ramanujan, and S.~Saurabh.
\newblock Solving \emph{d-}{SAT} via backdoors to small treewidth.
\newblock In {\em Proceedings of the Twenty-Sixth Annual {ACM-SIAM} Symposium
  on Discrete Algorithms, {SODA} 2015, San Diego, CA, USA, January 4-6, 2015},
  pages 630--641. {SIAM}, 2015.

\bibitem{FominV12}
F.~V. Fomin and Y.~Villanger.
\newblock Treewidth computation and extremal combinatorics.
\newblock {\em Combinatorica}, 32(3):289--308, 2012.

\bibitem{Freuder85}
E.~C. Freuder.
\newblock A sufficient condition for backtrack-bounded search.
\newblock {\em J. {ACM}}, 32(4):755--761, 1985.

\bibitem{GanianRamanujanSzeider16}
R.~Ganian, M.~S. Ramanujan, and S.~Szeider.
\newblock Discovering archipelagos of tractability for constraint satisfaction
  and counting.
\newblock In {\em Proceedings of the Twenty-Seventh Annual {ACM-SIAM} Symposium
  on Discrete Algorithms, {SODA} 2016, Arlington, VA, USA, January 10-12,
  2016}, pages 1670--1681, 2016.

\bibitem{GaspersMisraOrdyniakSzeiderZivny14}
S.~Gaspers, N.~Misra, S.~Ordyniak, S.~Szeider, and S.~Zivny.
\newblock Backdoors into heterogeneous classes of {SAT} and {CSP}.
\newblock In C.~E. Brodley and P.~Stone, editors, {\em Proceedings of the
  Twenty-Eighth {AAAI} Conference on Artificial Intelligence, July 27 -31,
  2014, Qu{\'{e}}bec City, Qu{\'{e}}bec, Canada.}, pages 2652--2658. {AAAI}
  Press, 2014.

\bibitem{GaspersOrdyniakRamanujanSaurabhSzeider13}
S.~Gaspers, S.~Ordyniak, M.~S. Ramanujan, S.~Saurabh, and S.~Szeider.
\newblock Backdoors to q-{H}orn.
\newblock In N.~Portier and T.~Wilke, editors, {\em 30th International
  Symposium on Theoretical Aspects of Computer Science, STACS 2013, February 27
  - March 2, 2013, Kiel, Germany}, volume~20, pages 67--79, 2013.

\bibitem{GaspersSzeider13}
S.~Gaspers and S.~Szeider.
\newblock Strong backdoors to bounded treewidth {SAT}.
\newblock In {\em 54th Annual IEEE Symposium on Foundations of Computer
  Science, FOCS 2013, 26-29 October, 2013, Berkeley, CA, USA}, pages 489--498.
  IEEE Computer Society, 2013.

\bibitem{GottlobLeoneScarcello02b}
G.~Gottlob, N.~Leone, and F.~Scarcello.
\newblock Hypertree decompositions and tractable queries.
\newblock {\em J. of Computer and System Sciences}, 64(3):579--627, 2002.

\bibitem{Grohe07}
M.~Grohe.
\newblock The complexity of homomorphism and constraint satisfaction problems
  seen from the other side.
\newblock {\em J. of the ACM}, 54(1), 2007.

\bibitem{GroheKMW11}
M.~Grohe, K.~Kawarabayashi, D.~Marx, and P.~Wollan.
\newblock Finding topological subgraphs is fixed-parameter tractable.
\newblock In {\em Proceedings of the 43rd {ACM} Symposium on Theory of
  Computing, {STOC} 2011, San Jose, CA, USA, 6-8 June 2011}, pages 479--488,
  2011.

\bibitem{HellNesetril90}
P.~Hell and J.~Ne{\v{s}}et{\v{r}}il.
\newblock On the complexity of {$H$}-coloring.
\newblock 48(1):92--110, 1990.

\bibitem{HellNesetril08}
P.~Hell and J.~Nesetril.
\newblock Colouring, constraint satisfaction, and complexity.
\newblock {\em Computer Science Review}, 2(3):143--163, 2008.

\bibitem{HemaspaandraWilliams12}
L.~A. Hemaspaandra and R.~Williams.
\newblock {SIGACT} news complexity theory column 76: an atypical survey of
  typical-case heuristic algorithms.
\newblock {\em SIGACT News}, pages 70--89, 2012.

\bibitem{KawarabayashiT11}
K.~Kawarabayashi and M.~Thorup.
\newblock The minimum k-way cut of bounded size is fixed-parameter tractable.
\newblock In {\em {IEEE} 52nd Annual Symposium on Foundations of Computer
  Science, {FOCS} 2011, Palm Springs, CA, USA, October 22-25, 2011}, pages
  160--169, 2011.

\bibitem{Kolaitis03}
P.~G. Kolaitis.
\newblock Constraint satisfaction, databases, and logic.
\newblock In {\em IJCAI-03, Proceedings of the Eighteenth International Joint
  Conference on Artificial Intelligence, Acapulco, Mexico, August 9-15, 2003},
  pages 1587--1595. Morgan Kaufmann, 2003.

\bibitem{KolaitisVardi00}
P.~G. Kolaitis and M.~Y. Vardi.
\newblock Conjunctive-query containment and constraint satisfaction.
\newblock {\em J. of Computer and System Sciences}, 61(2):302--332, 2000.
\newblock Special issue on the Seventeenth ACM SIGACT-SIGMOD-SIGART Symposium
  on Principles of Database Systems (Seattle, WA, 1998).

\bibitem{Lagergren98}
J.~Lagergren.
\newblock Upper bounds on the size of obstructions and intertwines.
\newblock {\em J. Comb. Theory, Ser. {B}}, 73(1):7--40, 1998.

\bibitem{Marx13}
D.~Marx.
\newblock Tractable hypergraph properties for constraint satisfaction and
  conjunctive queries.
\newblock {\em J. of the ACM}, 60(6):Art. 42, 51, 2013.

\bibitem{Montanari74}
U.~Montanari.
\newblock Networks of constraints: fundamental properties and applications to
  picture processing.
\newblock {\em Information Sciences}, 7:95--132, 1974.

\bibitem{NishimuraRagdeSzeider04-informal}
N.~Nishimura, P.~Ragde, and S.~Szeider.
\newblock Detecting backdoor sets with respect to {Horn} and binary clauses.
\newblock In {\em Proceedings of SAT 2004 (Seventh International Conference on
  Theory and Applications of Satisfiability Testing, 10--13 May, 2004,
  Vancouver, BC, Canada)}, pages 96--103, 2004.

\bibitem{RazgonOSullivan09}
I.~Razgon and B.~O'Sullivan.
\newblock Almost 2-{SAT} is fixed parameter tractable.
\newblock {\em J. of Computer and System Sciences}, 75(8):435--450, 2009.

\bibitem{SamerSzeider10}
M.~Samer and S.~Szeider.
\newblock Algorithms for propositional model counting.
\newblock {\em J. Discrete Algorithms}, 8(1):50--64, 2010.

\bibitem{SamerSzeider10a}
M.~Samer and S.~Szeider.
\newblock Constraint satisfaction with bounded treewidth revisited.
\newblock {\em J. of Computer and System Sciences}, 76(2):103--114, 2010.

\bibitem{Schaefer78}
T.~J. Schaefer.
\newblock The complexity of satisfiability problems.
\newblock In {\em Conference Record of the Tenth Annual ACM Symposium on Theory
  of Computing (San Diego, Calif., 1978)}, pages 216--226. ACM, 1978.

\bibitem{WilliamsGomesSelman03}
R.~Williams, C.~Gomes, and B.~Selman.
\newblock Backdoors to typical case complexity.
\newblock In G.~Gottlob and T.~Walsh, editors, {\em Proceedings of the
  Eighteenth International Joint Conference on Artificial Intelligence, IJCAI
  2003}, pages 1173--1178. Morgan Kaufmann, 2003.

\bibitem{WilliamsGomesSelman03a}
R.~Williams, C.~Gomes, and B.~Selman.
\newblock On the connections between backdoors, restarts, and heavy-tailedness
  in combinatorial search.
\newblock In {\em Informal Proc. of the Sixth International Conference on
  Theory and Applications of Satisfiability Testing, S. Margherita Ligure -
  Portofino, Italy, May 5-8, 2003 (SAT 2003)}, pages 222--230, 2003.

\end{thebibliography}

\end{document}